%% file: 00_main.tex
\documentclass[acmsmall,nonacm]{acmart}

\pdfoutput=1  

\input{macros}

\setcopyright{acmlicensed}
\copyrightyear{2024}
\acmYear{2024}
\acmDOI{XXXXXXX.XXXXXXX}

\acmJournal{TALG}
\acmVolume{XXX}
\acmNumber{YYY}
\acmArticle{111}
\acmMonth{1}

\begin{document} 

\title{Competitive Data-Structure Dynamization}%
\titlenote{A conference version of this paper appeared in SODA 2021~\cite{mathieu2021competitive}.
The journal version was published in ACM Transactions on Algorithms,
\url{https://doi.org/10.1145/3672614} .}



\author{Claire Mathieu}
\email{Claire.Mathieu@irif.fr}
\affiliation{%
  \institution{CNRS}
  \city{Paris}
  \country{France}
}
\orcid{0000-0002-0517-112X}

\author{Rajmohan Rajaraman}
\affiliation{%
  \institution{Northeastern University}
  \city{Boston}
  \state{Massachusetts}
  \country{USA}
}
\email{rraj@ccs.neu.edu}
\authornote{Supported by NSF grants 1535929 and 1909363.}
\orcid{0009-0005-3610-9918}

\author{Neal E.~Young}
\affiliation{%
  \institution{University of California Riverside}
  \city{Riverside}
  \state{California}
  \country{USA}
}
\orcid{0000-0001-8144-3345}
\email{neal.young@ucr.edu}
\authornote{Supported by Google Research Award \& NSF grant 1619463.}

\author{Arman Yousefi}
\affiliation{%
 \institution{Google}
 \city{Los Angeles}
 \state{California}
 \country{USA}
}
\orcid{0000-0002-8760-539X}
\email{armany@google.com}



\input{05_abstract}


\received{1 January 1960}
\received[revised]{1 January 1960}
\received[accepted]{1 January 1960}

\setcopyright{rightsretained}
\acmJournal{TALG}
\acmYear{2024} \acmVolume{1} \acmNumber{1} \acmArticle{1} \acmMonth{1}\acmDOI{10.1145/3672614}

\maketitle

\input{10_introduction}

\setcounter{theorem}{0} 
\input{20_min_sum}

\input{30_k_component}

\input{40_offline}

\input{50_conclusion}
\bibliographystyle{ACM-Reference-Format}
\bibliography{zotero-auto-generated,additional-bib}

\input{90_appendix}
\end{document}

%% file: macros.tex
\usepackage{tikz}               
\usetikzlibrary{%
  calc,%
  intersections,%
  decorations.pathmorphing,%
  fadings,%
  decorations.pathreplacing,%
} 
\usepackage{forest}  

\usepackage{adjustbox} 

\usepackage{relsize} 

\usepackage{pgfplots}           
\pgfplotsset{compat=1.17}

\usepackage{amsmath}  

\usepackage{amsthm} 

\usepackage{xspace}  

\usepackage{graphicx}  

\usepackage{enumitem} 

\makeatletter
\NewEnviron{restateableTheorem}[1]{%
  \protected@write\@auxout{}{%
    \string\@restatetheorem{#1}{\detokenize\expandafter{\BODY}}%
  }%
  \begin{theorem}\label{#1}\BODY\end{theorem}%
}
\newcommand{\@restatetheorem}[2]{%
  \expandafter\gdef\csname restatethis@#1\endcsname{#2}%
}
\newcommand{\restateTheorem}[2]{%
  \begingroup
  \renewcommand{\thetheorem}{\ref{#1}}%
  \begin{theorem}[#2] \csname restatethis@#1\endcsname\end{theorem}%
  \endgroup
}

\NewEnviron{restateableCorollary}[1]{%
  \protected@write\@auxout{}{%
    \string\@restatecorollary{#1}{\detokenize\expandafter{\BODY}}%
  }%
  \begin{corollary}\label{#1}\BODY\end{corollary}%
}
\newcommand{\@restatecorollary}[2]{%
  \expandafter\gdef\csname restatethis@#1\endcsname{#2}%
}
\newcommand{\restateCorollary}[2]{%
  \begingroup
  \renewcommand{\thecorollary}{\ref{#1}}%
  \begin{corollary}[#2] \csname restatethis@#1\endcsname\end{corollary}%
  \endgroup
}
\makeatother

\AtEndPreamble{%
  \theoremstyle{acmplain}
  \newtheorem{sublemma}{Sublemma}
  \newtheorem{observation}[theorem]{Observation}
  \theoremstyle{acmdefinition}
  \newtheorem{definitions}[theorem]{Definitions}
}

\newenvironment{labeledProof}[1][]{\begin{proof}[#1]}{\end{proof}}

\newcommand{\junk}[1]{}

\newcommand{\E}{\operatorname{E}}
\newcommand{\giv}{\,|\,}

\newcommand{\Rp}{\mathbb{R}_{+}}

\newcommand{\eps}{\varepsilon}

\newcommand{\calN}{\mathcal N}

\newcommand{\mysf}[1]{\ensuremath{\operatorname{\sf #1}}}
\newcommand{\mymathname}[1]{\mysf{#1}}

\newcommand{\dsop}[1]{\textsf{#1}}

\newcommand{\wt}{\mymathname{wt}}
\newcommand{\credit}{\mymathname{credit}}
\newcommand{\cost}{\mymathname{cost}}
\newcommand{\OPT}{\mymathname{OPT}\xspace}
\newcommand{\nonred}[1]{\mymathname{nonred}(#1)}

\newcommand{\interval}[2]{[#1, #2]}

\newcommand{\inp}{I}
\newcommand{\inputs}[2]{\mathcal I_{#1}^{#2}}
\newcommand{\out}{\mathcal C}
\newcommand{\univ}[1]{{U_{#1}}}

\newcommand{\ListLengths}{\setlength{\itemsep}{0ex}\setlength{\topsep}{0ex}\setlength{\partopsep}{0ex}}
\newcounter{listCounter}

\newlist{steps}{enumerate}{7}
\setlist[steps]{wide,
  topsep=1pt,
  parsep=2pt,
  partopsep=1pt,
  itemsep=0pt,
  align=left,
  labelindent=0pt,
  itemindent=0pt,
  leftmargin=0pt, 
  rightmargin=0pt,
  labelsep=4pt,
  labelwidth=0pt,
}
\setlist[steps,1]{
  leftmargin=6pt, 
  label=\small\arabic*.,
  ref=\arabic*
}
\setlist[steps,2]{
  label=\small\arabic{stepsi}.\arabic*.,
  ref=\arabic{stepsi}.\arabic*
}
\setlist[steps,3]{
  label=\small\arabic{stepsi}.\arabic{stepsii}.\arabic*.,
  ref=\arabic{stepsi}.\arabic{stepsii}.\arabic*
}
\setlist[steps,4]{
  label=\small\arabic{stepsi}.\arabic{stepsii}.\arabic{stepsiii}.\arabic*.,
  ref=\arabic{stepsi}.\arabic{stepsii}.\arabic{stepsiii}.\arabic*
}
\setlist[steps,5]{
  label=\small\arabic{stepsi}.\arabic{stepsii}.\arabic{stepsiii}.\arabic{stepsiv}.\arabic*.,
  ref=\arabic{stepsi}.\arabic{stepsii}.\arabic{stepsiii}.\arabic{stepsiv}.\arabic*
}
\setlist[steps,6]{
  label=\small\arabic{stepsi}.\arabic{stepsii}.\arabic{stepsiii}.\arabic{stepsiv}.\arabic{stepsv}.\arabic*.,
  ref=\arabic{stepsi}.\arabic{stepsii}.\arabic{stepsiii}.\arabic{stepsiv}.\arabic{stepsv}.\arabic*
}
\setlist[steps,7]{
  label=\small\arabic{stepsi}.\arabic{stepsii}.\arabic{stepsiii}.\arabic{stepsiv}. \arabic{stepsv}. \arabic{stepsvi}.\arabic*.,
  ref=\arabic{stepsi}.\arabic{stepsii}.\arabic{stepsiii}.\arabic{stepsiv}.\arabic{stepsv}. \arabic{stepsvi}.\arabic*
}

\newcommand{\step}[1][]{\item\maybeLabel{#1}}

\makeatletter
\newcommand{\skipitem}{\refstepcounter{\@enumctr}}
\makeatother

\newcommand{\maybeLabel}[1]{\ifthenelse{\equal{#1}{}}{}{\label{#1}}}

\newcommand{\mycomment}[1]{\hfill{\footnotesize\emph{#1}}}

\newlength{\boxUnit}
\newlength{\boxPadding}
\newcommand{\iTemplate}[3]{
  \trimbox{-\boxPadding 0pt -\boxPadding 0pt}{%
    \begin{tikzpicture}[anchor=base, baseline]
      \node[#1,
        thin,
        text=black,
        text width={#2\boxUnit-\boxPadding-\boxPadding},
        inner sep=0mm,
        outer sep=0mm,
        minimum height={3ex},
        align=center,
        ] at (0,0) {\strut \hfill $ #3 $ \hfill \strut};
    \end{tikzpicture}%
  }}
\newcommand{\iBox}[2]{\iTemplate{draw,color=gray,fill=gray!10}{#1}{#2}}
\newcommand{\iGhost}[2]{\iTemplate{draw,color=gray!40}{#1}{#2}}
\newcommand{\iFix}[2]{\iTemplate{}{#1}{#2}}

\newcommand{\closedopen}[1]{[#1)}


%% file: 05_abstract.tex
\begin{abstract}
  \emph{Data-structure dynamization}
  is a general approach for making static data structures dynamic. 
  It is used extensively in geometric settings
  and in the guise of so-called \emph{merge (or compaction)} policies
  in big-data databases
  such as LevelDB and Google Bigtable.
  Previous theoretical work is based on worst-case analyses
  for uniform inputs --- insertions of one item at a time and non-varying read rate.
  In practice, merge policies must not only handle batch insertions and varying read/write ratios,
  they can take advantage of such non-uniformity to reduce cost on a per-input basis.

  To model this,
  we initiate the study of data-structure dynamization through the lens of competitive analysis,
  via two new online set-cover problems.
  For each, the input is a sequence of disjoint sets of weighted items.
  The sets are revealed one at a time.
  The algorithm must respond to each with a set cover
  that covers all items revealed so far.
  It obtains the cover incrementally from the previous cover
  by adding one or more sets
  and optionally removing existing sets.
  For each new set the algorithm incurs  \emph{build cost}
  equal to the weight of the items in the set.
  In the first problem
  the objective is to minimize total build cost
  plus total \emph{query cost}, 
  where the algorithm incurs a query cost at each time $t$
  equal to the current cover size.
  In the second problem,
  the objective is to minimize the build cost
  while keeping the query cost from exceeding $k$ (a given parameter) at any time.
  We give deterministic online algorithms for both variants,
  with competitive ratios of $\Theta(\log^* n)$ and $k$, respectively.
  The latter ratio is optimal for the second variant.
\end{abstract}

%% file: 10_introduction.tex
\section{Introduction}\label{sec: introduction}

\subsection{Background}\label{sec: background}
A \emph{static} data structure is built once to hold a fixed set of items, queried any number of times, and then destroyed,
without changing throughout its lifespan.
\emph{Dynamization} is a generic technique for transforming any static container data structure
into a \emph{dynamic} one
that supports
insertions and queries intermixed arbitrarily.
%
  (Deletions and updates can be supported
  as described in Section~\ref{sec: intro k component}.)
%
The dynamic structure holds all items inserted so far in a collection of static containers.
Insertions are supported by adding new static containers and deleting old ones.
Queries are supported by querying all (current) static containers.
Static containers are called \emph{components}.
Dynamization has been applied in computational geometry~\cite
{lee_computational_1984,
  chiang_dynamic_1992,
  agarwal_framework_2001,
  agarwal_approximating_2004,
  bronnimann_inplace_2004},
in geometric streaming algorithms~\cite
{har-peled_coresets_2004,
  bagchi_deterministic_2007, 
  feldman_turning_2013},
and to design external-memory dictionaries~\cite
{arge_oefficient_2004,
  vitter_algorithms_2008,
  aggarwal_hierarchical_1987,
  bender_cacheoblivious_2007}.

Bentley's \emph{binary transform}~\cite
{bentley_decomposable_1979,
  bentley_decomposable_1980},
later called the \emph{logarithmic method}~\cite
{vanleeuwen_art_1981,
  overmars_design_1987},
is a widely used example.
It maintains the invariant that the number of items in each component
is a distinct power of two.
Each \dsop{insert} operation mimics a binary increment:
it destroys the components of size $2^0, 2^1, 2^2,\ldots, 2^{j-1}$,
where $j\ge 0$ is the minimum such that there is no component of size $2^j$,
and builds one new component of size $2^{j}$,
holding the contents of the destroyed components and the inserted item.
%
(See Figure~\ref{fig: binary transform}.)
%
Meanwhile, each \dsop{query} operation queries all current components,
combining the results appropriately for the data type.
During $n$ insertions,
whenever an item is incorporated into a new component,
the item's new component is at least twice as large as its previous component,
so the item is in at most $\log_2 n$ component builds.
That is, the worst-case \emph{write amplification} is at most $\log_2 n$.
Meanwhile, the number of components never exceeds $\log_2 n$,
so each query examines at most $\log_2 n$ components.
That is, the worst-case \emph{read amplification} is at most $\log_2 n$.

\newcommand{\STACK}[2]{
  \begin{array}[b]{c}
    \begin{array}{|@{\,~}c@{~\,}|} \hline
      \mathbf #2 \\
      \hline
    \end{array}
    \\ \raisebox{-8pt}{\small $t=#1$}
  \end{array}
}
\newcommand{\STACKa}[2]{\STACK{#1}{#2}}
\newcommand{\STACKb}[3]{\STACKa{#1}{#2 \\\hline #3}}
\newcommand{\STACKc}[4]{\STACKb{#1}{#2}{#3 \\\hline #4}}
\newcommand{\SEP}{\raisebox{12pt}{${\!\rightarrow\!}$}}
\begin{figure}[t]
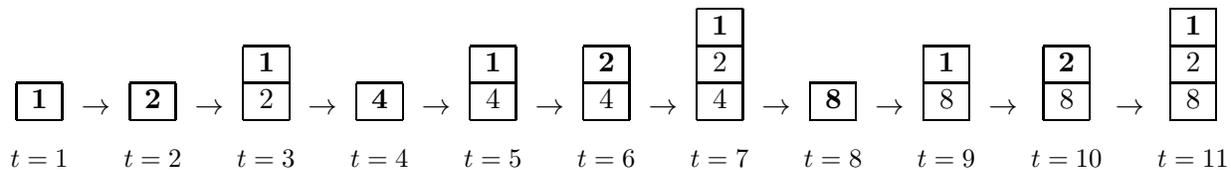

  \begin{displaymath}
    \STACKa 1 1
    \SEP
    \STACKa 2 2
    \SEP
    \STACKb 3 1 2
    \SEP
    \STACKa 4 4
    \SEP
    \STACKb 5 1 4
    \SEP
    \STACKb 6 2 4
    \SEP
    \STACKc 7 1 2 4
    \SEP
    \STACKa 8  8
    \SEP
    \STACKb 9 1 8 
    \SEP
    \STACKb {10} 2  8
    \SEP
    \STACKc {11} 1 2  8
  \end{displaymath}    
  \caption{Steps
    1--11
    of the binary transform~\cite{bentley_decomposable_1979,bentley_decomposable_1980}.
    Each cell \boxed{\,i\,}\ is a component holding $i$ items,
    where $i$ is a distinct power of two.
    In each step one item is inserted
    and held in the new (top, bolded) component.
  }\label{fig: binary transform}
\end{figure}

Bentley and Saxe's \emph{$k$-binomial} transform is a variant of the binary transform~\cite
{bentley_decomposable_1980}.
It maintains $k$ components at all times,
of respective sizes $\binom{i_1} 1$,  $\binom{i_2} 2$, \ldots, $\binom{i_k} k$
such that $0\le i_1 < i_2 < \cdots < i_k$.
(This decomposition is guaranteed to exist and to be unique.
Figure~\ref{fig: binomial transform} gives an example.) 
It thus ensures read amplification at most $k$, independent of $n$,
but its write amplification is at most $(k!\, n)^{1/k}$, about $\frac{k}{e} n^{1/k}$ for large $k$.
This tradeoff between \emph{worst-case} read amplification and \emph{worst-case} write amplification
%
is optimal up to lower-order terms, as is the tradeoff achieved by the binary transform
  (see Section~\ref{subsec: offline}).

These worst-case guarantees on read- and write-amplification
hold both for \emph{uniform} inputs
(where the inserted items have roughly the same sizes,
and insertions and queries occur at uniform and balanced rates)
and for non-uniform inputs.
But a non-uniform input can be substantially easier,
in that it admits a solution with average write amplification (over all inserted items)
and average read amplification (over all queries)
well below worst case,
achieving lower total cost.
(Roughly, this is achieved 
by trading build cost for query cost
as the read/write ratio varies.
For intuition consider a long sequence of insertions followed by a long sequence of queries.)
Worst-case dynamization analyses do not capture this.
Indeed, transforms such as those above
do not adapt to non-uniformity.
Their build and query costs are close to worst case even on non-uniform inputs.

We propose two new dynamization problems---\emph{Min-Sum Dynamization} and \emph{$k$-Component Dynamization}---that model non-uniform insertions and queries.
We consider these as online problems
and use competitive analysis to measure how well algorithms for them take advantage of non-uniformity.
We introduce new algorithms that have substantially better competitive ratios
than existing algorithms.

\newcommand{\STACKA}[2]
{\STACK{#1}{#2}}
%
\begin{figure}[t]
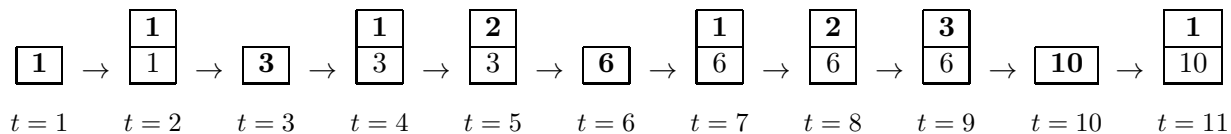

  \begin{displaymath}
    \STACKA 1 1
    \SEP
    \STACKb 2 1 1
    \SEP
    \STACKA 3 3
    \SEP
    \STACKb 4 1 3
    \SEP
    \STACKb 5 2 3
    \SEP
    \STACKA 6 6
    \SEP
    \STACKb 7 1 6
    \SEP
    \STACKb 8 2 6
    \SEP
    \STACKb 9 3 6
    \SEP
    \STACKA {10} {{10}}
    \SEP
    \STACKb {11} 1 {10}
  \end{displaymath}    
  \caption{
    Steps
    1--11
    of the $2$-binomial transform~\cite{bentley_decomposable_1980}.
    At time $t$
    the top and bottom components hold 
    $\binom{i_1} 1$ and $\binom{i_2} 2$ items where $0\le i_1<i_2$
    and $\binom{i_1} 1 + \binom{i_2} 2 = t$. For example at time t = 8, $i_1 = 2$ and $i_2 = 4$.
    If $i_1=0$ there is only one component, the bottom component.
  }\label{fig: binomial transform}
\end{figure}%

\paragraph{Relevance to industrial LSM systems.}
Dynamization algorithms underlie standard implementations of external-memory (i.e., disk-based) ordered dictionaries,
where they are called~\emph{merge} (or \emph{compaction}) policies~\cite{luo_lsmbased_2018}.
Recently inserted key/value pairs are cached in RAM to the extent possible,
while older pairs are stored in immutable (static) on-disk files (the components).
Each \dsop{query} (if not resolved in cache)
searches the current components for the queried key,
using one disk access\footnote
{Database servers are typically configured so that RAM size is 1--3\% of
  disk size, even as RAM and disk sizes grow according to Moore's law~\cite[p.~227]{graefe_modern_2010}.
  A disk block typically holds at least thousands of items.
  Hence, an index for every disk component,
  storing the minimum item in each disk block in the component,
  fits easily in RAM.  Then querying any component (a file storing its items in sorted order)
  for a given item  requires accessing just one disk block,
  determined by checking the index~\cite[p.~232]{graefe_modern_2010}. }
per component.
The components are managed using the merge policy:
periodically, the cached pairs are flushed to disk in one batch,
which is treated as an inserted item
and incorporated by building and deleting components\footnote
{Crucially, builds use sequential (as opposed to random) disk access.
  This is why these systems outperform B$^+$ trees on write-heavy workloads.
  See~\cite[\S\,2.2.1--2.2.2]{luo_lsmbased_2018} for details.}
according to the policy.
The build cost captures the time building on-disk components,
while the query cost captures the time responding to queries.

O'Neil et al's seminal \emph{log-structured merge} (LSM) architecture~\cite{oneil_logstructured_1996}
(building on~\cite{severance_differential_1976,rosenblum_design_1991})
was one of the first to adapt dynamization to external-memory dictionaries
as described above.
Its dynamization scheme can be viewed as a parameterized generalization of Bentley's binary transform.
As the parameter varies,
the tradeoff it achieves between read amplification and write amplification
is optimal (in some parameter regimes)
among all external-memory structures~\cite
{arge_external_2002, brodal_lower_2003,yi_dynamic_2012}.

Many subsequent and current industrial systems---including so-called
NoSQL and NewSQL databases---have LSM architectures.
These include 
Google's Bigtable~\cite{chang_bigtable_2008} (and Spanner~\cite{corbett_spanner_2013}),
Amazon's Dynamo~\cite{decandia_dynamo_2007},
Accumulo (by the NSA)~\cite{kepner_achieving_2014}, 
AsterixDB~\cite{alsubaiee_asterixdb_2014},
Facebook's Cassandra~\cite{lakshman_cassandra_2010}, 
HBase and Accordion (used by Yahoo! and others)~\cite{george_hbase_2011,bortnikov_accordion_2018}, 
LevelDB~\cite{dent_getting_2013}, 
and RocksDB~\cite{dong_optimizing_2017}.

Non-uniform inputs can be particularly important in production LSM systems,
where the sizes of inserted batches
can vary by orders of magnitude~\cite[\S 2]{bortnikov_accordion_2018}
(see also~\cite{brodal_cacheaware_2005,bender_cacheadaptive_2014,bender_closing_2020})
and the query and insertion rates can vary substantially with time. 
As discussed previously, such non-uniform workloads can have optimal cost well below the worst-case cost.
Industrial compaction policies do adapt to non-uniformity, but only heuristically.
Bigtable's default compaction policy (which, like the $k$-binomial transform, is configured by a single parameter $k$ and maintains at most $k$ components) is as follows:
in response to each \dsop{insert} (cache flush), create a new component holding the inserted items;
then, if there are more than $k$ components,
merge the $i$ most-recently created components into one,
where $i \ge 2$ is chosen minimally so
that, for each remaining component $S$, the size of $S$ in bytes
exceeds the total size of all components newer than $S$~\cite{staelin}.
Both the worst-case build cost and the competitive ratio of this algorithm are suboptimal.

\subsection{Problem definitions}\label{sec: definitions} 
The definitions of the two dynamization problems below model insertions and queries.
The end of the section gives generalizations that allow updates, deletions, and item expiration
as implemented (lazily) in typical LSM systems.

Recall that a \emph{set cover} of a given set $S$ of items
is a collection of subsets whose union is $S$.

%

%

\begin{definitions}
  \begin{description}
  \item[]
    An \emph{input}  is a sequence $\inp=(\inp_1, \inp_2, \ldots, \inp_n)$
    of pairwise-disjoint sets of weighted items.
    Each item $x\in\inp_t$ is said to be ``\emph{inserted at time $t$}''.
    The weight of each item $x$,
    denoted $\wt(x)$,
    must be non-negative.

  \item[]
    A \emph{solution} is a sequence $\out=(\out_1, \out_2, \ldots, \out_n)$,
    where each $\out_t$ is a set cover for the items inserted by time $t$.
    That is, $\bigcup_{S\in \out_t} S = \univ t$,
    where $\univ t = \bigcup_{i=1}^t \inp_i$.
    The sets in each $\out_t$ are called \emph{components}.

  \item[]
    The \emph{build cost} at time $t$
    is the total weight in \emph{new} sets, that is
    $\sum_{S\in \out_t\setminus \out_{t-1}} \wt(S)$,
    where $\wt(S)$ denotes $\sum_{x\in S} \wt(x)$.
    (For time $t=1$ we define $\out_0$ to be the empty set.)

  \item[]
    The \emph{query cost} at time $t$ is $|\out_t|$, that is, the number of components in the current cover, $\out_t$.

  \item[]
    Given an input, the objective of the \emph{Min-Sum Dynamization} problem
    is to find a solution of minimum \emph{total cost} (the sum of all build costs and query costs
    over time).

  \item[]
    The objective of the \emph{$k$-Component Dynamization} problem
    is to find a solution having minimum total build cost,
    among solutions with maximum query cost $k$ (that is, $\max_t |\out_t| \le k$).

  \item[]
    An 
    algorithm is
    \emph{online} if for every input $\inp$ it outputs a solution $\out$ such that
    at each time $t$ its cover $\out_t$ is independent of $\inp_{t+1}, \inp_{t+2}, \ldots, \inp_n$,
    all build costs $\wt_{t'}(S)$ at times $t'> t$, and $n$.

  \item[]
    An algorithm's \emph{competitive ratio}, $c^*(m)$, 
    is the supremum,
    over all inputs with $m$ non-empty insertions,
    of the cost of the algorithm's solution divided by the optimum cost for the input.

  \item[]
    An algorithm is \emph{$c(m)$-competitive} if its competitive ratio $c^*(m)$ is at most $c(m)$ for all $m$.
  \end{description}
\end{definitions}


\paragraph{Remarks for Min-Sum Dynamization.}
The definition of total read cost
(as $\sum_{t=1}^n |\out_t|$) models, a-priori, exactly one query per insert.
This keeps the problem statement relatively simple.
However, applications can have any number of queries per insert.
This can be modeled by reduction.
To model consecutive queries with no intervening insertions,
separate the consecutive queries by artificial insertions with $\inp_t=\emptyset$
(inserting an empty set).
To model consecutive insertions with no intervening query,
aggregate the consecutive insertions into a single insertion.

Note that uniformly scaling item weights
changes build cost relative to query cost.
In LSM applications, each unit of query cost represents the time for one \emph{random-access} disk read,
whereas each unit of build cost represents the (much smaller, amortized) time per byte 
during \emph{sequential} disk reads and writes.
To model these costs,
take the weight of each item $x$ to be its size in bytes,
times the time per byte for a sequential disk read and write,
divided by the (much larger) time for one random-access disk read.

\paragraph{Remark for $k$-Component Dynamization.}
Among well-studied problems,
\emph{Dynamic TCP Acknowledgment}~\cite
{karlin_dynamic_2003, buchbinder_design_2009},
a generalization of the classic ski-rental problem,
is perhaps technically closest to $k$-Component Dynamization.
TCP Acknowledgment
can be viewed as a continuous-time variant of $2$-Component Dynamization
in which
building a new component that contains all items inserted so far
(corresponding to a ``TCP-ack'')
has cost 1 (regardless of the component weight).

\paragraph{Deletions, updates, and expiration.}
The problems as defined above model queries and insertions.
Next we extend the definitions to allow modelling updates, deletions, and item expiration
as they typically happen (lazily) in LSM dictionaries.

In this context we assume each item is a weighted key/value pair, timestamped by insertion time,
and possibly having an expiration time.
(The item weight is typically proportional to the size in bytes of the key/value pair.)
Updates and deletions are lazy (``out of place''~\cite[\S 2]{luo_lsmbased_2018},~\cite{lim_accurate_2016}):
\dsop{update} just inserts an item with the given key/value pair (as usual),
while \dsop{delete} inserts an item for the given key with a so-called \emph{tombstone} (a.k.a.~\emph{antimatter}) value.
Multiple items with the same key may be stored (possibly in multiple components), but only the newest matters:
a query, given a key, returns the newest item inserted for that key,
or ``none'' if that item is a tombstone or has expired.
When a component $S$ is built, it is ``garbage collected'':
for each key,
among the items in $S$ with that key,
only the newest is written to disk---all others are discarded.

To model this, we define three generalizations of the problems.
To keep the definitions clean,
in each variant the input sets must still be disjoint
and the current cover must still contain all items inserted so far.
To model updates, deletions, and expirations
we only redefine the build cost.

\begin{definitions}
  \begin{description}
  \item[\emph{Decreasing Weights.}]
    Each item $x\in\inp_t$ has weights $\wt_{t}(x) \ge \wt_{t+1}(x) \ge \cdots \ge \wt_n(x)$.
    The cost of building a component $S\subseteq \univ t$ at time $t$ is redefined as $\wt_{t}(S) = \sum_{x\in S} \wt_t(x)$.
    We use this variant as a stepping stone to the LSM variant, next.
    
  \item[\emph{LSM.}]
    Each item is a timestamped key/value pair with an expiration time.
    Given a subset $S$ of items,
    the set of \emph{non-redundant} items in $S$,
    denoted $\nonred S$,
    consists of those that have no newer item in $S$ with the same key.
    The cost of building a component $S$ at time $t$, denoted $\wt_t(S)$,
    is redefined as the sum, over all non-redundant items $x$ in $S$,
    of the item weight $\wt(x)$, or the weight of the tombstone item for $x$
    if $x$ has expired.
    The latter weight must be at most $\wt(x)$.
    Items with the same key may have different weights, must have distinct timestamps,
    and can occur in different components.
    For any two items $x\in \inp_t$ and $x'\in\inp_{t'}$ with $t<t'$,
    the timestamp of $x$ must be less than the timestamp of $x'$.
    This variant applies to LSM systems.\footnote
    {LSM systems delete tombstone items during full merges
      (i.e., when building a component $S=\univ t$ at time $t$).
      This is not captured by the LSM model here,
      but is captured by the following \emph{general} model.
      See Section~\ref{sec: variations}.}

  \item[\emph{General.}]
    Instead of weighting the items,
    build costs are specified directly for sets.
    At each time $t$ a \emph{build-cost function} $\wt_t{:~}2^{\univ t}\rightarrow \Rp$ is revealed (along with $\inp_t$),
    directly specifying the build cost $\wt_t(S)$ for every possible component $S\subseteq\univ t$.
    The build-cost function must obey
    the following restrictions,
    for all times $i\le t$ and sets $S,\,S' \subseteq \univ t$:
    \begin{description}
    \item[~(R1) sub-additivity:] $\wt_t(S\cup S') \le \wt_t(S) + \wt_t(S')$.
      (The cost of building a component holding the union of two sets
      is at most the combined cost of building two components that hold the respective sets.)

    \item[~(R2) suffix monotonicity:]~~if
      $S\ne\univ t$, then $\wt_t(S \setminus \univ i) \le \wt_t(S)$, 
      (The cost of building a component holding a set $S$ of items
      is at least the cost of building a component holding just those items in $S$
      that were inserted after time $i$.
      The exception for $S=\univ t$ allows modeling full removal of tombstone items during full merges.)

    \item[~(R3) temporal monotonicity:]~~$\wt_i(S) \ge \wt_t(S)$
      (The cost of building a component to hold $S$ does not increase over time.
      Note, for example, that in the LSM model item expirations can cause the cost to decrease over time.)
    \end{description}
  \end{description}
\end{definitions}

We chose Restrictions (R1)--(R3)
so that the resulting problem has several competing properties:
it should be relatively simple,
sufficiently general to model practical LSM systems,
and sufficiently restricted to allow competitive online algorithms.
The build costs implicit in the LSM and Decreasing Weights variants
do obey (R1)--(R3).\footnote
{The LSM build cost obeys (R1) because $\nonred {S\cup S'} \subseteq \nonred S \cup \nonred S'$.
  It obeys (R2) because $\nonred {S\setminus  \univ i} \subseteq \nonred S$.
  It obeys (R3) because the tombstone weight for each item $x$ is at most $\wt(x)$.}
The restrictions would also hold, for example, if each item had a fixed weight
and $\wt_t(S) = \max_{x\in S} \wt(x)$.

\subsection{Statement of results}\label{sec: results} 

\subsubsection*{Min-Sum Dynamization}

Recall that the \emph{iterated logarithm} (base $e$)
is the slowly growing function defined inductively by $\log_e^* m = 1 + \log_e^* \log_e m$,
with the base case $\log_e^* m = 0$ for $m \le 1$.
(Our analysis will use base $\sqrt 2$ instead of $e$.
Note that $\log^*_{\sqrt 2} m = \Theta(\log^*_e m)$,
so inside O-notation the base is omitted.)

\restateTheorem{thm: min-sum}{Section~\ref{sec: min-sum}}

%

Roughly speaking, every $2^j$ time steps ($j\in\{0,1,2,\ldots\}$),
the algorithm merges all components of weight $2^j$ or less into one.
Figure~\ref{fig: min-sum example} illustrates one execution of the algorithm.
The bound in the theorem is tight for the algorithm.  

In contrast, consider
the naive adaptation of Bentley's binary transform
(i.e., treat each insertion $\inp_t$ as a size-1 item,
then apply the transform).
On inputs with $\wt(\inp_t) = 1$ for all $t$
the algorithms produce the same (optimal) solution.
But, as we show in Lemma~\ref{lemma: min-sum binary transform} in the appendix, the competitive ratio of the naive adaptation is $\Omega(\log n)$.
%
%

Min-Sum Dynamization is a special case of \emph{Set Cover with Service Costs},
for which Buchbinder et al.{} give a randomized online algorithm~\cite
{buchbinder2014competitive}.
For Min-Sum Dynamization, their bound on the algorithm's competitive ratio
simplifies to $O(\log^2 n)$.

\subsubsection*{$K$-Component Dynamization and its generalizations}\label{sec: intro k component}\label{sec: k-component results}
\restateTheorem{thm: k-component lower bound}{Section~\ref{sec: k-component lower bound}}


\restateTheorem{thm: k-component decreasing}{Section~\ref{sec: k-component decreasing}}

For comparison, consider the naive adaptation of Bentley and Saxe's $k$-binomial transform
to $k$-Component Dynamization
(treat each insertion $\inp_t$ as one size-1 item, then apply the transform).
On inputs with $\wt(\inp_t) = 1$ for all $t$,
the two algorithms produce essentially the same optimal solution.
%
  But, as we show in Lemma~\ref{lemma: naive bs non-competitive} in the appendix,
  the competitive ratio of the naive adaptation is $\Omega(k n^{1/k})$ for any $k\ge 2$.

Bigtable's default algorithm (Section~\ref{sec: background}) solves
$k$-Component Dynamization, but its competitive ratio is $\Omega(n)$.
For example, with $k=2$, given an instance with $\wt(\inp_1) = 3$, $\wt(\inp_2) = 1$,
and $\wt(\inp_t) = 0$ for $t\ge 3$, it pays $n+2$,
while the optimum is $4$.
(In fact, the algorithm is memoryless ---
each $\out_t$ is determined by $\out_{t-1}$ and $\inp_t$.
No deterministic memoryless algorithm has competitive ratio independent of $n$.)
Even for uniform instances ($\wt(\inp_t) = 1$ for all $t$),
Bigtable's default incurs cost quadratic in $n$,
whereas the optimum is $\Theta(k n^{1+1/k})$.


Bentley and Saxe showed that their solutions
were optimal (for uniform inputs)
among a restricted class of solutions
that they called \emph{arboreal transforms}~\cite{bentley_decomposable_1980}.
Here we call such solutions \emph{newest-first}:

\begin{definition}\label{def: newest-first}\label{def: lightest-first}
  A solution $\out$ is \emph{newest-first}
  if at each time $t$,
  if $\inp_t=\emptyset$ it creates no new components,
  and otherwise it creates one new component,
  by merging $\inp_t$ with some $i\ge 0$ newest components into a single component
  (destroying the merged components).
  Likewise, $\out$ is \emph{lightest-first} if, at each time $t$ with $\inp_t\ne \emptyset$,
  it merges $\inp_t$ with some $i\ge 0$ lightest components.
  An \emph{algorithm} is newest-first (lightest-first)
  if it produces only newest-first (lightest-first) solutions.
\end{definition}

The Min-Sum Dynamization algorithm Adaptive-Binary (Figure~\ref{fig: min-sum alg}) is lightest-first.  
The $k$-Component Dynamization algorithm Greedy-Dual (Figure~\ref{fig: k-component decreasing}) is newest-first.  
In a newest-first solution, every cover $\out_t$ partitions the set $\univ t$ of current items
into components of the form $\bigcup_{h=i}^j \inp_h$ for some $i, j$.

Any newest-first algorithm for the decreasing-weights variant of either problem
can be ``bootstrapped'' into
an equally good algorithm for the LSM variant:

\restateTheorem{thm: meta}{Section~\ref{sec: k-component LSM}}

With Theorem~\ref{thm: k-component decreasing} this gives the following corollary:

\restateCorollary{cor: k-component LSM}{Section~\ref{sec: k-component LSM}}


Finally we give an algorithm for the general variant:

\restateTheorem{thm: k-component general}{Section~\ref{sec: k-component general}}

The algorithm $\mysf B_k$
partitions the input sequence into phases.
Before the start of each phase, it has just one component in its cover,
called the current ``root'',
containing all items inserted before the start of the phase.
During the phase, $\mysf B_k$ recursively simulates $\mysf B_{k-1}$
to handle the insertions occurring during the phase,
and uses the cover that consists of the root component together
with the (at most $k-1$) components currently used by $\mysf B_{k-1}$.
At the end of the phase, $\mysf B_k$ does a \emph{full merge}
--- it merges all components into one new component, which becomes the  new root.
It extends the phase maximally subject to the constraint that
the cost incurred by $\mysf B_{k-1}$ during the phase does not exceed
$k-1$ times the cost of the full merge that ends the phase.

\subsection{Properties of optimal offline solutions}\label{subsec: offline}

Bentley and Saxe showed that, among newest-first solutions (which they called \emph{arboreal}),
their various transforms were near-optimal for uniform inputs~\cite
{bentley_decomposable_1979,bentley_decomposable_1980}.
Mehlhorn showed (also for uniform inputs) that the best
newest-first solutions have cost at most a constant times optimum~\cite{mehlhorn_lower_1981}.
We generalize and strengthen Mehlhorn's result:

\restateTheorem{thm: offline}{Section~\ref{sec: offline}}

One consequence is that Bentley and Saxe's transforms
give optimal solutions (up to lower-order terms)  for uniform inputs.
Another is that, for Min-Sum and $k$-Component Dynamization,
optimal solutions can be computed
in time $O(n^3)$ and $O(k n^3)$, respectively,
because optimal newest-first solutions can be computed in these time bounds
via natural dynamic programs.

The body of the paper gives the proofs of Theorems~\ref{thm: min-sum}--\ref{thm: offline}.


%% file: 20_min_sum.tex
\section{Min-Sum Dynamization (Theorem~\ref{thm: min-sum})}\label{sec: min-sum}
\begin{restateableTheorem}{thm: min-sum}
  For Min-Sum Dynamization, the online algorithm Adaptive-Binary (Figure~\ref{fig: min-sum alg}) has competitive ratio
  $\Theta(\log^* m)$, where $m\le n$ is the number of non-empty insertions.
\end{restateableTheorem}
%
We prove the theorem in two parts:
%
\begin{enumerate}
\item[\em(i)]  The competitive ratio is $O(\log^*m)$ (proof in Section~\ref{sec: min-sum upper}).
\item[\em(ii)] The competitive ratio is $\Omega(\log^*m)$ (proof in Section~\ref{sec: min-sum lower}).
\end{enumerate}


\begin{figure*}[t]
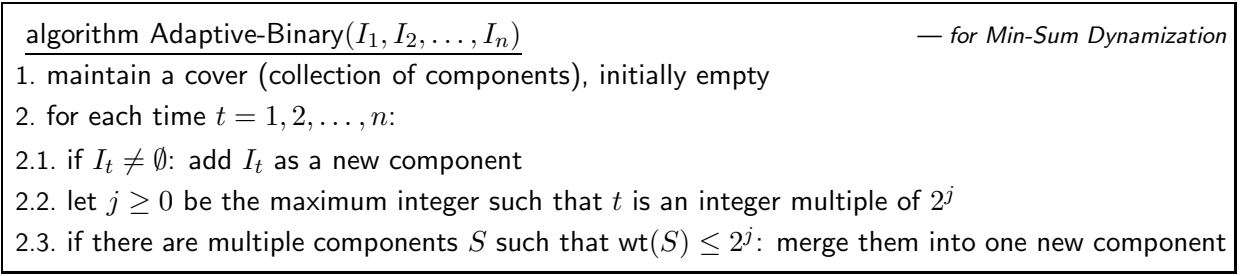

  \framebox{\parbox{0.98\textwidth}{\sf\smaller[0]\setlength{\parskip}{1.5pt}

      \begin{steps}
      \item[] \underline{algorithm Adaptive-Binary$(\inp_1, \inp_2, \ldots, \inp_n)$}
        \mycomment{---for Min-Sum Dynamization}
        
        \step maintain a cover (collection of components), initially empty
        
        \step for each time $t=1,2,\ldots, n$:
        \begin{steps}
          \step\label{line: inp set}
          if $\inp_t\ne\emptyset$: add $\inp_t$ as a new component 

          \step\label{line: capacity}
          let $j\ge 0$ be the maximum integer such that $t$ is an integer multiple of $2^j$
          \mycomment{--well defined as $2^0=1$}

          \step\label{line: min-sum merge}
          if there are multiple components $S$ such that $\wt(S) \le 2^j$:
          \begin{steps}
            \step: merge them into one new component
          \end{steps}
        \end{steps}
      \end{steps}
    }}
  \caption{
    A  $\Theta(\log^* m)$-competitive algorithm for Min-Sum Dynamization
    (Theorem~\ref{thm: min-sum}).}\label{fig: min-sum alg} 
\end{figure*}

%
%

\subsection{Part (i): the competitive ratio is $O(\log^* m)$}\label{sec: min-sum upper}

\input{21_min_sum_ub}

\subsection{Part (ii): the competitive ratio is $\Omega(\log^* m)$}\label{sec: min-sum lower}

\input{22_min_sum_lb}

Note that in Lemma~\ref{lemma: min-sum lower}, $n\approx m^2$, so $\log^* m =\Omega(\log^* n)$.
The upper bound in Section~\ref{sec: min-sum upper} and the lower bound in Lemma~\ref{lemma: min-sum lower}
prove Theorem~\ref{thm: min-sum}.


%% file: 21_min_sum_ub.tex
Fix an input $\inp=(\inp_1, \inp_2, \ldots, \inp_n)$ with $m\le n$ non-empty sets.
Let $\out$ be the algorithm's solution.
Let $\out^*$ be an optimal solution, of cost \OPT.
For any time $t$,
call the $2^j$ chosen in Line~\ref{line: capacity} the \emph{capacity} $\mu(t)$ of time $t$,
and let $S_t$ be the newly created component (if any) in Line~\ref{line: min-sum merge}.

It is convenient to over-count the algorithm's build cost as follows.
In Line~\ref{line: min-sum merge},
if there is exactly one component $S$ with $\wt(S) \le 2^j$,
the algorithm as stated does not change the current cover,
but we pretend for the analysis that it does ---
specifically, that it destroys and rebuilds $S$, paying its build cost $\wt(S)$ again at time $t$.
This allows a clean statement of the next lemma.
In the remainder of the proof, the ``build cost'' of the algorithm refers to this over-counted build cost. 

We first bound the total query cost, $\sum_t |\out_t|$, of $\out$.
\begin{lemma}\label{lemma: min-sum query}
  The total query cost of $\out$ is at most twice the (over-counted) build cost of $\out$, plus $\OPT$.
\end{lemma}
\begin{proof}
  %
  Consider the components with weight less than 1.
  %
  By inspection of the algorithm each cover $\out_t$ has at most one such component --- the component $S_t$ created at time $t$.
  Therefore, the query cost from the components with weight less than 1 is at most $n$.

  It remains to consider the components with weight at least 1.
  Let $S$ be any component in $\out$ of weight $\wt(S) \ge 1$.
  Each new occurrence of $S$ in $\out$
  contributes at most $2\wt(S)$ to $\out$'s query cost.
  Indeed, let $2^j\ge \wt(S)$ be the next larger power of 2.
  Times with capacity $2^j$ or more occur every $2^j$ time steps.
  So, after $\out$ creates $S$, $\out$ destroys $S$ within $2^j \le 2\wt(S)$ time steps; note that we are using here the over-counted build cost.
  So $\out$'s query cost from such components is at most twice the build cost of $\out$.

  Thus, the total query cost from all components is at most twice the build cost of $\out$ plus $n$.
  The lemma follows since the query cost of $\out^*$ is at least $n$, so $n \le \OPT$.
\end{proof}

Define $\Delta$ to be the maximum number of components merged by the algorithm in response to any query.
Note that $\Delta \le m$ simply because there are at most $m$
components at any given time in $\out$.
(Only Line~\ref{line: inp set} increases the number of components, and it does so only if $\inp_t$ is non-empty.)
The remainder of the section bounds the build cost of $\out$ by $O(\log^*(\Delta) \OPT)$.
By Lemma~\ref{lemma: min-sum query}, this will imply prove Part (i) of the theorem.

The total weight of all components $\inp_t$ that the algorithm creates in Line~\ref{line: inp set}
is $\sum_t \wt(\inp_t)$, which is at most \OPT because every $x\in\inp_t$
is in at least one new component in $\out^*$ (at time $t$).
To finish, we bound the (over-counted) build cost of the components
that the algorithm builds in Line~\ref{line: min-sum merge}, i.e., $\sum_t \wt(S_t)$.

\begin{observation}\label{obs: separated}
  The difference between any two distinct times $t$ and $t'$  is at least $\min\{\mu(t), \mu(t')\}$.
\end{observation}
(This holds because $t$ and $t'$ are distinct integer multiples of $\min\{\mu(t), \mu(t')\}$.
See Figure~\ref{fig: capacities}.)

\begin{figure}
  \centering
  \input{23_min_sum_capacities}
  \caption{The capacities $\mu(t)$ as a function of $t$.}\label{fig: capacities}
\end{figure}

\paragraph{Charging scheme.}
For each time $t$ at which Line~\ref{line: min-sum merge} creates a new component $S_t$,
have $S_t$ \emph{charge} to each item $x\in S_t$ the weight $\wt(x)$ of $x$.
Have $x$ in turn charge $\wt(x)$ to each optimal component $S^*\in \out^*_t$
that contains $x$ at time $t$.
The entire build cost $\sum_t \wt(S_t)$ is charged to components in $\out^*$.
To finish, we show that each component $S^*$ in $\out^*$
is charged $O(\log^* \Delta)$ times $S^*$'s contribution
(via its build and query costs) to \OPT.


  Throughout, 
  given integer times $t$ and $t'$,
  let $\interval t {t'}$ denote
  the \emph{(time) interval} $\{t, t+1, \ldots, t'\}$.
  (This is non-standard notation.)

Fix any such $S^*$.
Define $\interval {t_1} {t_2}$ to be the interval of
$S^*$ in $\out^*$.
That is, $\out^*$ adds $S^*$ to its cover at time $t_1$,
where it remains through time $t_2$,
so its contribution to \OPT is $t_2-t_1+1 + \wt(S^*)$.
At each (integer) time $t\in\interval {t_1} {t_2}$, component $S^*$ is charged $\wt(S^*\cap S_t)$.
To finish, we show 
$\sum_{t=t_1}^{t_2} \wt(S^*\cap S_t) = O(t_2-t_1 +  \log^*(\Delta)\wt(S^*))$.


By Observation~\ref{obs: separated},
there can be at most one time $t'\in\interval {t_1} {t_2}$ with capacity $\mu(t') > t_2-t_1+1$.
If there is such a time $t'$, the charge received then, i.e.~$\wt(S^*\cap S_{t'})$, is at most $\wt(S^*)$.
To finish, we bound the charges at the times $t\in\interval {t_1} {t_2} \setminus \{t'\}$,
with $\mu(t) \le t_2 - t_1 + 1$.

\begin{definition}[dominant]
  Classify each such time $t$ and $\out$'s component $S_t$ as \emph{dominant}
  if the capacity $\mu(t)$ strictly exceeds
  the capacity $\mu(i)$ of every earlier time $i\in \interval {t_1} {t-1}$
  ($\mu(t) > \max_{i=t_1}^{t-1} \mu(i)$) in $S^*$'s interval $\interval {t_1} {t_2}$.
  Otherwise $t$ and $S_t$ are \emph{non-dominant}.
\end{definition}

\begin{lemma}[non-dominant times]\label{lemma: min-sum non-dominant}
  The net charge to $S^*$ at non-dominant times is at most $t_2 - t_1$.
\end{lemma}
\begin{proof}
  Let $\tau_1$ be any dominant time.
  Let $\tau_2>\tau_1$ be the next larger dominant time step, if any, else $t_2+1$.
  Consider the charge to $S^*$ during the open interval $(\tau_1, \tau_{2})$.
  We show that this charge is at most $\tau_{2} - \tau_1 - 1$.

  Component $S^*$ is built at time $t_1\le \tau_1$,  so $S^* \subseteq \univ {\tau_1}$.
  %
  At time $\tau_1$, every item $x$ that can charge $S^*$
    (that is, $x\in S^*$) is in some component $S$ in $\out_{\tau_1}$.
    By the definition of dominant,
    each time in $t\in (\tau_1, \tau_2)$ has capacity $\mu(t) \le \mu(\tau_1)$, since otherwise $\tau_2$ would not be the \emph{next} dominant time.
    So, the components $S$ in $\out_{\tau_1}$ that have weight $\wt(S) > \mu(\tau_1)$
    remain unchanged in $\out$ throughout $(\tau_1, \tau_2)$,
    and the items in them do not charge $S^*$ during $(\tau_1, \tau_2)$.
%
  So we need only consider items in components $S$ in $\out_{\tau_1}$ with $\wt(S) \le\mu(\tau_1)$.
  Assume there are such components.
  By inspection of the algorithm, there can only be one: the component $S_{\tau_1}$ built at time $\tau_1$.
  All charges in $(\tau_1, \tau_2)$ come from items $x\in S_{\tau_1}\cap S^*$.
  
  Let $\tau_1 = t'_1 < t'_2 < \cdots < t'_\ell$ be the times in $\closedopen{\tau_1, \tau_2}$ when these items
  are put in a new component.
  These are the times in $(\tau_1, \tau_2)$ when $S^*$ is charged,
  and, at each, the charge is $\wt(S^*\cap S_{\tau_1}) \le \wt(S_{\tau_1})$,
  so the total charge to $S^*$ during $(\tau_1, \tau_2)$ is at most $(\ell-1) \wt(S_{\tau_1})$.

  At each time $t'_i$ with $i\ge 2$ the previous component $S_{t'_{i-1}}$, 
  of weight at least $\wt(S_{\tau_1})$, is merged.
  So each time $t'_i$ has capacity $\mu(t'_i) \ge \wt(S_{\tau_1})$.
  By Observation~\ref{obs: separated},
  the difference between each time $t'_i$ and the next $t'_{i+1}$ is at least $\wt(S_{\tau_1})$.
  So $(\ell-1) \wt(S_{\tau_1}) \le t'_{\ell} - t'_1 \le \tau_2 - \tau_1 - 1$.

  By the two previous paragraphs the charge to $S^*$
  during $(\tau_1, \tau_2)$ is at most $\tau_2 - \tau_1 - 1$.
  Summing over the dominant times $\tau_1$ in $[t_1, t_2]$ proves the lemma.
\end{proof}

Let $D$ be the set of dominant times.
For the rest of the proof all times that we consider are dominant.
Note that all times that are congested or uncongested (as defined next) are dominant. 

\newcommand{\congest}{\kappa}
\begin{definition}[congestion]
  For any time $t\in D$ and component $S_t$,
  define the \emph{congestion} of $t$ and $S_t$ to be $\wt(S_t\cap S^*)/\mu(t)$, 
  the amount $S_t$ charges $S^*$, divided by the capacity $\mu(t)$.
  Call $t$ and $S_t$ \emph{congested} if this congestion exceeds
  $\congest$,
  and \emph{uncongested} otherwise ($\kappa > 0$ is a constant that is specified later).
\end{definition}

\begin{lemma}[uncongested times]\label{lemma: min-sum dominant uncongested}
  The total charge to $S^*$ at uncongested times is $O(t_2-t_1)$.
\end{lemma}
\begin{proof}
  The charge to $S^*$ at any uncongested time $t$ is at most
  %
  $\congest \mu(t)$,
  %
  so the total charge to $\out^*$ during such times is at most $\congest\sum_{t\in D} \mu(t)$.
  By definition of \emph{dominant},
  the capacity $\mu(t)$ for each $t\in D$ is a distinct power of 2 no larger than $t_2-t_1+1$.
  So $\sum_{t\in D} \mu(t)$ is at most $2(t_2-t_1+1)$,
  and the total charge to $\out^*$
  during uncongested times is $O(t_2-t_1)$.
\end{proof}

\newcommand{\cng}{\alpha}

\begin{lemma}[congested times]\label{lemma: min-sum dominant congested}
  The total charge to $S^*$ at congested times
  is $O(\wt(S^*) \log^* \Delta)$.
\end{lemma}
\begin{proof}
  Let $Z$ denote the set of congested times. 
  For each item $x\in S^*$,
  let $W(x)$ be the collection of congested components that contain $x$ and charge $S^*$.
  The total charge to $S^*$ at congested times is $\sum_{x\in S^*} |W(x)| \wt(x)$.

  To bound this, we use a random experiment that starts by
  choosing a random item $X$ in $S^*$,
  where each item $x$ has probability proportional to $\wt(x)$ of being chosen:
  $\Pr[X = x] = \wt(x)/\wt(S^*)$.

  We will show that $\E_X[|W(X)|]$ is $O(\log^* \Delta)$.
  Since $\E_X[|W(X)|] = \sum_{x\in S^*} |W(x)| \wt(x)/\wt(S^*)$,
  this will imply that the total charge is $O(\log^* \Delta) \wt(S^*)$, proving the lemma.
  
  \paragraph{The merge forest for $S^*$.}
  Define the following \emph{merge forest}.
  There is a leaf $\{x\}$ for each item $x\in S^*$.
  There is a non-leaf node $S_t$ for each congested component $S_t$.
  The parent of each leaf $\{x\}$ is the first congested component $S_t$ that contains $x$
  (that is, $t=\min\{i \in Z : x\in S_i$), if any.
  The parent of each node $S_t$ is the next congested component $S_{t'}$ that contains
  all items in $S_t$ (that is, $t' = \min\{i \in Z: i > t,\, S_t \subseteq S_{i}\}$), if any.
  Parentless nodes are roots.

  The random walk starts at the root of the tree that holds leaf $\{X\}$,
  then steps along the path to that leaf in the tree.
  In this way it traces (in reverse) the sequence $W(X)=\{S_i : X\in S_i\}$ of congested components
  that $X$ entered during $\interval {t_1} {t_2}$.
  The number of steps is $|W(X)|$.
  To finish, we show that the expected number of steps is $O(\log^*\Delta)$.
  
  Each non-leaf node $S_t$ in the tree has congestion $\wt(S_t\cap S^*)/\mu(t)$,
  which is at least $\congest$ and at most $\Delta$.
  For the proof, define the congestion of each leaf $x$ to be $2^\Delta$.
  To finish, we argue that \emph{with each step of the random walk,
    the iterated logarithm of the current node's congestion
    increases in expectation by at least $1/5$.}

  \paragraph{A step in the random walk.}
  Fix any non-leaf node $S_t$.
  Let $\cng_t = \wt(S_t\cap S^*)/\mu(t)$ be its congestion.
  The walk visits $S_t$
  with probability $\wt(S^*\cap S_t)/\wt(S^*)$.
  Condition on this event (that is, $X \in S_t$).
  Let random variable $\cng'$ be the congestion of the child of $S_t$ next visited.
  \begin{sublemma}\label{sublemma: bound}
    For any $\beta \in \closedopen{\cng_t, 2^{\Delta}}$,
    $\Pr[\cng' > \beta \giv X\in S_t]$
    is at least $1-\cng_t^{-1}(2 + \log_2 \beta)$.
  \end{sublemma}
  \begin{proof}
    Consider any child $S_{t'}$ of $S_t$ with $\cng_{t'} \le \beta$.
    We will bound the probability that $S_{t'}$ is visited next (i.e., $X\in S_{t'}$).
    Node $S_{t'}$ is not a leaf, as $\cng_{t'} < 2^{\Delta}$.
    Define $j(t')$ so that
    its capacity $\mu(t')$ equals $\mu(t)/2^{j(t')}$.
    (That is, $j(t') = \log_2(\mu(t)/\mu(t'))$.)
    The definitions and $\cng_{t'} \le \beta$ imply
    \begin{align}
          \Pr[X\in S_{t'} \giv X\in S_t]
          ~=~
          \frac{\wt(S_{t'}\cap S^*)}{\wt(S_t\cap S^*)}
              ~=~
              \frac{\alpha_{t'}\, \mu(t')}{\cng_t\, \mu(t)}
              ~\le~
              \frac{\beta\, \mu(t)/ 2^{j(t')}}{\cng_t\, \mu(t)}
              ~=~
              \frac{\beta}{\cng_t\, 2^{j(t')}}. \label{obs: two}
    \end{align}

    Also, the algorithm merged a component containing $S_{t'}$ at time $t$,
    so $\wt(S_{t'})\le \mu(t)$,
    so
    \begin{align}
          \Pr[X\in S_{t'} \giv X\in S_t]
          ~=~
          \frac{\wt(S_{t'}\cap S^*)}{\wt(S_t\cap S^*)}
              ~=~
              \frac{\wt(S_{t'}\cap S^*)}{\cng_t \, \mu(t)}
              ~\le~
              \frac{\wt(S_{t'})}{\cng_t \, \mu(t)}
              ~\le~
              \frac{1}{\cng_t}. \label{obs: one}
    \end{align}

    Combining Bounds~\eqref{obs: two} and~\eqref{obs: one},
    $\Pr[X\in S_{t'} | X\in S_t]$ is at most $\cng_t^{-1}\min(1, \beta\, 2^{-j(t')})$.
    Summing this bound over all children $S_{t'}$ of $S_t$ with congestion $\cng_{t'} \le \beta$,
    and using that each $j(t')$ is a distinct positive integer,
    the probability that $\cng' \le \beta$ is at most
    \[
      \cng_t^{-1} \sum_{j = 1}^\infty \min(1, \beta\, 2^{-j})
      \le
      \cng_t^{-1} \int_0^\infty \min(1, \beta\, 2^{-j}) \, dj
      = \cng_t^{-1}(\log_2(\beta) + 1/ \ln 2)
    \]
    (splitting the integral at $j=\log_2 \beta$).
    The sublemma follows from $1/\ln 2 \le 2$.
  \end{proof}
  
  Next we lower-bound the expected increase in the $\log^*$ of the congestion in this step.
  We use $\sqrt 2$ as the base of the iterated log.\footnote
  {Defined by $\log_{\sqrt 2}^*\cng_t = 0$ if $\cng_t \le 8$, else $1+\log_{\sqrt 2}^*(\log_{\sqrt 2} \cng_t)$.
    Note that $\log^*_{\sqrt 2} \cng_t = \Theta(\log^*_e \cng_t)$.}
  Then $\log^*(2^{\cng_t/2}) = 1+ \log^* \cng_t$,
  so, conditioned on $X\in S_t$,
  \[
    \E[\log^* \cng'] 
    \,\ge\, \Pr[\cng' \ge \cng_t ] \log^* \cng_t
    \,+\,\Pr[\cng' \ge 2^{\cng_t/2}].
    \]
  Bounding the two probabilities above
  via Sublemma~\ref{sublemma: bound}
  with $\beta=\cng_t$ and $\beta=2^{\cng_t/2}$,
  the right-hand side above is 
  \begin{align*}
    & {}\ge
      [1 - \cng_t^{-1}(2 + \log_2\cng_t)]\log^* \cng_t
      \,+\,
      [1 - \cng_t^{-1}(2  + \cng_t/2)] \\
    & {}=
      \log^*(\cng_t) + 1/2 - [2 + (2 + \log_2 \cng_t)\log^*\cng_t]/\cng_t \\
    & {}\ge \log^*(\cng_t) + 1/2 - 3/10 ~=~ \log^*(\cng_t) + 1/5,
  \end{align*}
  where the last inequality follows from $\cng_t \ge \congest$ ($t$ is congested) and by setting $\congest \ge 64$.
  It follows that $E[\log^*\alpha' - \log^*\alpha_t \giv X \in S_t] \ge 1/5$.
  That is, in each step, the expected increase in the iterated logarithm of the congestion is at least 1/5.

  Let random variable $L=|W(X)|$ be the length of the random walk.
  Let random variable $\alpha'_i$ be the congestion of the $i$th node on the walk.
  By the previous section, for each $i$, given that $i < L$,
  $\E[\log^* \alpha'_{i+1} - \log^* \alpha'_i \giv \alpha'_i] \ge 1/5$.
  It follows by
  %
  Wald's equation
  (see~\cite[p.~370]{borodin1998online} and~\cite[Lemma~4.1]{young2000kmedians})
  that
  $\E[\log^* \alpha'_L - \log^* \alpha'_1] \ge \E[L]/5$.
  Since $\alpha'_L = 2^\Delta$ and $\log^* \alpha'_1 \ge 0$,
  we have $\E[\log^* \alpha'_L - \log^*\alpha'_1] \le \log^* 2^\Delta$.
  %
  It follows that $\E[L] \le 5 \log^* 2^\Delta$.
  Recall that the base of the iterated logarithm is $\sqrt{2}$; so, $\log^* 2^{\Delta} = 2 + \log^*\Delta$, yielding $\E[L] \le 10 + 5\log^* \Delta$.
  %
  That is, the expected length of the random walk is $O(\log^*\Delta)$.
  By the discussion at the start of the proof, this implies the lemma.
\end{proof}

To recap,
for each component $S_t$ built by the algorithm,
the (over-counted) build cost
is charged item by item to those components in the optimal solution $\out^*$ that currently contain the item.
In this way, the algorithm's total over-counted build cost $\sum_t \wt(S_t)$
is charged to components in $\out^*$.
By Lemmas~\ref{lemma: min-sum non-dominant}--\ref{lemma: min-sum dominant congested},
each component $S^*$ in the optimal solution $\out^*$
is charged $O(1)$ times its contribution $t_2-t_1$ to the query cost of $\out^*$
plus (in expectation) $O(\log^* m)$ times its contribution $\wt(S^*)$ to the build cost of $\out^*$.
It follows that the expected build cost incurred by the algorithm
is $O(\log^* m)$ times the cost of $\out^*$.

By Lemma~\ref{lemma: min-sum query}, the total \emph{query} cost incurred by the algorithm
is at most twice the algorithm's over-counted build cost plus the cost of $\out^*$.
It follows that the total (build and query) cost incurred by the algorithm
is $O(\log^* (m))$ times the cost of $\out^*$.
That is, the competitive ratio is $O(\log^* m)$, proving Part (i) of Theorem~\ref{thm: min-sum}.\footnote
{Curiously, the algorithm's cost
  is in fact $O(1)$ times the query cost of $\out^*$
  plus $O(\log^* m)$ times its build cost.}

%% file: 23_min_sum_capacities.tex
\hspace*{-0.1\columnwidth}\begin{tikzpicture}\footnotesize
\pgfplotstableread[row sep=\\,col sep=&]{
      t & mu \\
1 & 1 \\
2 & 2 \\
3 & 1 \\
4 & 4 \\
5 & 1 \\
6 & 2 \\
7 & 1 \\
8 & 8 \\
9 & 1 \\
10 & 2 \\
11 & 1 \\
12 & 4 \\
13 & 1 \\
14 & 2 \\
15 & 1 \\
16 & 16 \\
17 & 1 \\
    }\mydata
    \begin{axis}[
        ybar,
        xlabel={$t$},
        ylabel={$\mu(t)$},
        symbolic y coords={0,1,2,4,8,16,32}, 
        xtick=data,
        ymax=16,
        height=2.1in,
        width=3.7in,
        ]
      \addplot[black,fill=white!90!black] table[x=t,y=mu]{\mydata};
    \end{axis}
\end{tikzpicture}


%% file: 22_min_sum_lb.tex
  \begin{figure*}
  \centering
  \input{20_min_sum_example}
  \caption{\small
    The ``merge tree'' for an execution of the Adaptive-Binary algorithm
    (Figure~\ref{fig: min-sum alg}).
    The input sequence starts with $m=132$ \dsop{insert}s $\inp_1, \inp_2, \ldots, \inp_{132}$
    --- one for each leaf, of weight equal to leaf's label.
    It continues with $2^{16} - 132$ empty \dsop{insert}s ($\inp_t=\emptyset$).
    At each time $t=2^9, 2^{10}, 2^{11}, \ldots, 2^{17}$
    (during the empty \dsop{insert}s)
    the algorithm merges all components of weight $t$ to form a single new component, their parent.
    In this way, the algorithm builds a component for each node, with weight equal to the node's label.
    At time $t=2^{17}$ the final component is built---the root, of weight $2^{18}$, containing all items.
    The algorithm merges each item four times, so pays build cost  $4\times 2^{18}$.
  }\label{fig: min-sum example}
\end{figure*}

  \begin{figure*}
    \centering
    \input{22_min_sum_T_inf}
    \caption{
      The top three levels of $T_\infty$.  Each node $i$ has $2^{i-p(i)}$ children,
      where $p(i)$ is the parent of $i$ (exc.~$p(1)=0$).
      The merge tree $T^N_2$ (Figure~\ref{fig: min-sum example}) consists of these three levels,
      with each node $i$ given weight $2^{N-p(i)}$, so the nodes with weight $2^{N-i}$
      are the $2^{i-p(i)}$ children of node $i$, and their total weight equals the weight of node $i$.  Note that the label of a node in the merge tree is its weight, and the merge tree of Figure~\ref{fig: min-sum example} is $T^{18}_2$.
    }\label{fig: min-sum T}
  \end{figure*}
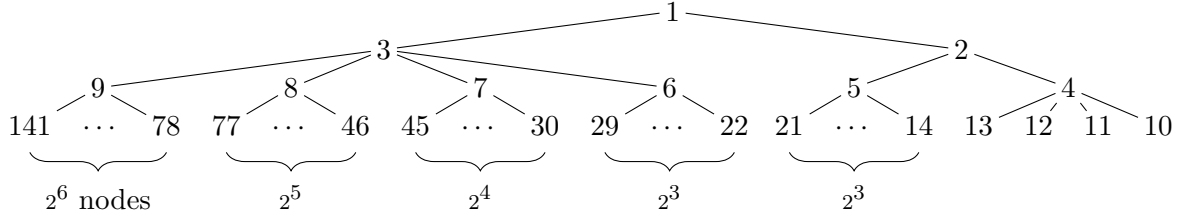

\begin{lemma}\label{lemma: min-sum lower}
  The competitive ratio of the Adaptive-Binary algorithm (Figure~\ref{fig: min-sum alg})
  is $\Omega(\log^* m)$.
\end{lemma}
\begin{proof}
  We will show a ratio of $\Omega(\log^* m)$ on a particular class of inputs,
  one for each integer $D\ge 0$.
  %
  (Figure~\ref{fig: min-sum example} describes the input $\inp$ for $D=2$ and the resulting merge tree, of depth $D+1$.)
  %
  
  \paragraph{The desired merge tree.}
  For reference,
  define an infinite rooted tree $T_\infty$ with node set $\{1, 2, 3, \ldots\}$ by the iterative process shown
  in Figure~\ref{fig: tree definition}.
  \begin{figure}
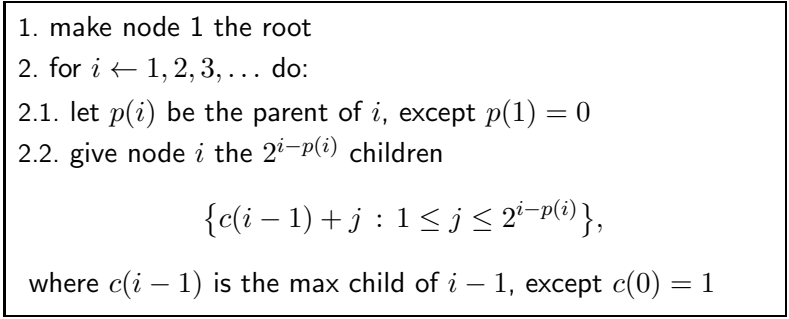
\centering
    \framebox{\parbox{0.8\textwidth}{
        \begin{steps}\sf
          \step make node 1 the root
          \step for $i \gets 1, 2, 3, \ldots$ do:
          \begin{steps}\label{block: split}
            \step let $p(i)$ be the parent of $i$, except $p(1) = 0$
            \step\label{line: allocate}
            give node $i$ the $2^{i - p(i)}$ children
            \(\big\{c({i-1})+j \,:\, 1 \le j \le 2^{i-p(i)}\big\}\),
            \\\hspace*{1em} where $c(i-1)$ is the max child of $i-1$, except~$c(0) = 1$
          \end{steps}
        \end{steps}
      }}
  \caption{An algorithm defining the tree $T_\infty$, with nodes $\{1, 2, 3, \ldots\}$.}\label{fig: tree definition}
  \end{figure}
  Each iteration $i$ defines the children of node $i$.
  %
  Node $i$ has $2^{i-p(i)}$ children, where $p(i)$ is the parent of $i$ (exc.\ $p(1) = 0$); these children are
  %
  allocated greedily from the ``next available'' nodes,
  so that each node $i\ge 2$ is given exactly one parent.
  The depth of $i$ is non-decreasing with $i$.\footnote
  {This follows by induction: Line~\ref{line: allocate}
    ensures that $p(i') \le p(i)$ for $i' < i$, so inductively
    $\text{depth}(i') = 1 + \text{depth}(p(i')) \le 1+ \text{depth}(p(i)) = \text{depth}(i)$.}

  %
  Figure~\ref{fig: min-sum T} shows the top three levels of $T_\infty$.
  %
  Let $n_d$ be the number of nodes of depth $d$ or less in $T_\infty$.
  Each such node $i$ satisfies $i\le n_d$
  (as depth is non-decreasing with $i$),
  so, inspecting Line~\ref{line: allocate}, node $i$ has at most $2^i \le 2^{n_d}$ children.
  Each node of depth $d+1$ or less is either the root or a child of a node of depth $d$ or less,
  so $n_{d+1} \le 1 + n_d 2^{n_d} \le 2^{2^{n_d}}$.
  Taking the $\log^*$ of both sides gives
  $\log^* n_{d+1} \le 2 + \log^* n_d$.
  Inductively, $\log^* n_d \le 2 d$ for each $d$.

  Now fix an integer $D\ge 0$.
  Define the desired merge tree, $T^N_D$, to be the subtree of $T_\infty$ induced by the nodes of depth at most $D+1$.
  Let $m$ be the number of leaves in $T^N_D$.
  By the previous paragraph (and $m\le n_{D+1}$), \emph{every leaf in $T^N_D$ has depth $\Omega(\log^* m)$}.

  Assign weights to the nodes in $T^N_D$ as
  follows.
  Fix $N=2\,n_{D}$.
  Give each node $i$ weight $2^{N-p(i)}$,
  where $p(i)$ is the parent of $i$ (except $p(1)=0$).
  Each weight is a power of two,
  and the nodes of any given weight $2^{N-i}$ are exactly the $2^{i-p(i)}$ children of node $i$.
  The weight of each parent $i$ equals the total weight of its children.
  
  \paragraph{The input.}
  Define the input $\inp$ as follows.
  For each time $t\in\{1,2,\ldots, m\}$, insert a set $\inp_t$
  containing just one item whose weight equals
  the weight of the $t$th leaf of $T^N_D$.
  Then, at each time $t\in\{m+1, m+2, \ldots, 2^{N-1}\}$,
  insert an empty set $\inp_t = \emptyset$.

  In the following, we place a lower bound on the cost of the algorithm on input $\inp$.
  For this, we establish a matching between the algorithm's cover and the leaves of $T^N_D$, which guides our bound on the build cost of the algorithm.

  \paragraph{No merges until last non-empty insertion.}
  %
  The algorithm does no merges before time $\min_{i=1}^m \wt(\inp_i)$, which is the minimum leaf weight in $T^N_D$.
  This is because if a merge occurs at a time $t$ then there must be more than one component of weight at most $t$ at that time
  (by step~\ref{line: min-sum merge} of the algorithm, see Figure~\ref{fig: min-sum alg}).
  The lightest leaves are the children of node $n_D$,
  of weight $2^{N-n_D}$.
  Since the total leaf weight is the weight of the root, $2^N$,
  it follows that $m 2^{N-n_D} \le 2^N$, that is, $m \le 2^{n_D} = 2^{N-n_D}$
  (using $N=2\,n_D$).
  So, \emph{the algorithm does no merges until time $t(n_D)=2^{N-n_D}$ (after all non-empty insertions).}

  \paragraph{The algorithm's merge tree matches $T^N_D$.}
  %
  By the previous two paragraphs,
  %
  just before time $t(n_D)=2^{N-n_D}$ 
  the algorithm's cover \emph{matches} the leaves of $T^N_D$,
  meaning that the cover's components correspond to the leaves,
  with each component weighing the same as its corresponding node.
  The leaves are $\{j : p(j) \le n_d < j\}$.
  So the following invariant holds initially, for $i=n_D$:
  
  \emph{For each $i\in\{n_D, n_D-1, \ldots, 2, 1\}$, just before time $t(i) = 2^{N-i}$,
    the algorithm's cover $\out_{t(i)}$  matches the nodes in $Q_i$, defined as}
  \[Q_i \doteq \{j : 2^{N-j} < t(i) \le 2^{N-p(j)}\} = \{j : p(j) \le i < j\}. \]
  Informally, these are the nodes $j$ that have not yet been merged by time $t(i)$,
  because their weight $2^{N-p(j)}$ is at least $t(i)$,
  but whose children (the nodes of weight $2^{N-j}$) if any, have already been merged.

  %
  We establish the invariant for all $i$ using a backward induction.  Assume inductively that the invariant holds for a given $i$.  We show it holds for $i-1$.
  %
  At time $t(i)$, the algorithm merges the components
  of weight at most $\mu(t(i)) = t(i) = 2^{N-i}$ in its cover.
  By the invariant, these are the components of weight $t(i) = 2^{N-i}$,
  corresponding to the children of node $i$ (which are all in $Q_i$).
  They leave the cover and are replaced by their union,
  whose weight equals $2^{N-p(i)}$.
  Likewise, by the definition (and $p(j) < j$)
  \[Q_{i-1} = \{i\} \cup Q_i \setminus \{ j : p(j) = i \}, \]
  so the resulting cover matches $Q_{i-1}$,
  with the new component corresponding to node $i$.
  The minimum-weight nodes in $Q_{i-1}$ are then $\{j : p(j) = i-1\}$,
  the children of node $i-1$.
  These have weight $2^{N-(i-1)} = t(i-1)$,
  so the algorithm keeps this cover until just before time $t(i-1)$,
  so that the invariant is maintained for $i-1$.

  Inductively, we arrive at the validity of the  invariant for $i=1$:
  just before time $t(1)=2^{N-1}=n$, 
  the algorithm's cover contains the components corresponding to
  $\{j : p(j) = 1 < j\}$, with weight $2^{N-p(j)} = 2^{N-1}=n$.
  At time $n$ they are merged form the final component of weight $2^N$,
  corresponding to the root node 1.
  So the algorithm's merge tree matches $T^N_D$.

  \paragraph{Competitive ratio.}
  Each leaf in the merge tree has depth $\Omega(\log^* m)$,
  so every item is merged $\Omega(\log^* m)$ times,
  and the algorithm's build cost is  $\Omega(\wt(1) \log^* m) = \Omega(n \log^* m)$
  (using $\wt(1)=2n$).

  But the optimal cost is $\Theta(n)$.
  (Consider the solution that merges all input sets into one component at time $m$,
  just after all non-empty insertions.
  Its query cost is $\sum_{t=1}^{m-1} t + \sum_{t=m}^n 1 =O(m^2 + n)$.
  Its merge cost is $2\wt(1) = O(n)$.
  Recalling that $m \le 2^{n_D} = 2^{N/2} = O(\sqrt n)$, the optimal cost is $O(n)$.)

  So the competitive ratio is $\Omega(\log^* m)$.
\end{proof}


%% file: 20_min_sum_example.tex
{
  \newcommand{\TT}[1]{{\scriptstyle 2}^{#1}}
  \begin{forest}
    s sep=0.45em, 
    for tree = {
      draw=black!40, 
      inner sep=1pt,
      text centered,
      draw=none, fill=none, 
      math content,
      l sep=3pt, l=0              
    }, 
    [{\TT {18}}
      [{\TT {17}}
        [{\TT {15}}
          [{\TT {9}}, name=e1]
          [\cdots, draw=none, fill=none, no edge]
          [{\TT {9}}, name=e2]
          ]
        [{\TT {15}}
          [{\TT {10}}, name=d1]
          [\cdots, draw=none, fill=none, no edge]
          [{\TT {10}}, name=d2]
          ]
        [{\TT {15}}
          [{\TT {11}}, name=c1]
          [\cdots, draw=none, fill=none, no edge]
          [{\TT {11}}, name=c2]
          ]
        [{\TT {15}}
          [{\TT {12}}, name=b1]
          [\cdots, draw=none, fill=none, no edge]
          [{\TT {12}}, name=b2]
          ]
        ]
      [{\TT {17}}
        [{\TT {16}}
          [{\TT {13}}, name=a1]
          [\cdots, draw=none, fill=none, no edge]
          [{\TT {13}}, name=a2]
          ]
        [{\TT {16}}
          [{\TT {14}}]
          [{\TT {14}}]
          [{\TT {14}}]
          [{\TT {14}}]
          ]
        ]
      ]
    \draw [decorate,decoration={brace,amplitude=7pt,mirror}] ([yshift=-3.5mm] a1.center) -- node[below=3mm]{ $\TT 3$} ([yshift=-3.5mm] a2.center); 
    \draw [decorate,decoration={brace,amplitude=7pt,mirror}] ([yshift=-3.5mm] b1.center) -- node[below=3mm]{ $\TT 3$} ([yshift=-3.5mm] b2.center); 
    \draw [decorate,decoration={brace,amplitude=7pt,mirror}] ([yshift=-3.5mm] c1.center) -- node[below=3mm]{ $\TT 4$} ([yshift=-3.5mm] c2.center); 
    \draw [decorate,decoration={brace,amplitude=7pt,mirror}] ([yshift=-3.5mm] d1.center) -- node[below=3mm]{ $\TT 5$} ([yshift=-3.5mm] d2.center); 
    \draw [decorate,decoration={brace,amplitude=7pt,mirror}] ([yshift=-3.5mm] e1.center) -- node[below=3mm]{ $\TT 6$ leaves} ([yshift=-3.5mm] e2.center); 
  \end{forest}
}

%% file: 22_min_sum_T_inf.tex
{
  \newcommand{\TT}[1]{{\scriptstyle 2}^{#1}}
  \newcommand{\XX}[1]{{#1}}
  \begin{forest}
    s sep=0.4em, 
    for tree = {
      draw=black!40, 
      inner sep=1pt,
      text centered,
      draw=none, fill=none, 
      math content,
      l sep=3pt, l=0              
    }, 
    [{\XX {1}}
      [{\XX {3}}
        [{\XX {9}}
          [{\XX {141}}, name=e1]
          [\cdots, draw=none, fill=none, no edge]
          [{\XX {78}}, name=e2]
          ]
        [{\XX {8}}
          [{\XX {77}}, name=d1]
          [\cdots, draw=none, fill=none, no edge]
          [{\XX {46}}, name=d2]
          ]
        [{\XX {7}}
          [{\XX {45}}, name=c1]
          [\cdots, draw=none, fill=none, no edge]
          [{\XX {30}}, name=c2]
          ]
        [{\XX {6}}
          [{\XX {29}}, name=b1]
          [\cdots, draw=none, fill=none, no edge]
          [{\XX {22}}, name=b2]
          ]
        ]
      [{\XX {2}}
        [{\XX {5}}
          [{\XX {21}}, name=a1]
          [\cdots, draw=none, fill=none, no edge]
          [{\XX {14}}, name=a2]
          ]
        [{\XX {4}}
          [{\XX {13}}]
          [{\XX {12}}]
          [{\XX {11}}]
          [{\XX {10}}]
          ]
        ]
      ]
    \draw [decorate,decoration={brace,amplitude=7pt,mirror}] ([yshift=-3.5mm] a1.center) -- node[below=3mm]{ $\TT 3$} ([yshift=-3.5mm] a2.center); 
    \draw [decorate,decoration={brace,amplitude=7pt,mirror}] ([yshift=-3.5mm] b1.center) -- node[below=3mm]{ $\TT 3$} ([yshift=-3.5mm] b2.center); 
    \draw [decorate,decoration={brace,amplitude=7pt,mirror}] ([yshift=-3.5mm] c1.center) -- node[below=3mm]{ $\TT 4$} ([yshift=-3.5mm] c2.center); 
    \draw [decorate,decoration={brace,amplitude=7pt,mirror}] ([yshift=-3.5mm] d1.center) -- node[below=3mm]{ $\TT 5$} ([yshift=-3.5mm] d2.center); 
    \draw [decorate,decoration={brace,amplitude=7pt,mirror}] ([yshift=-3.5mm] e1.center) -- node[below=3mm]{ $\TT 6$ nodes} ([yshift=-3.5mm] e2.center); 
  \end{forest}
}

%% file: 30_k_component.tex
\section{$K$-Component Dynamization and variants (Theorems~\ref{thm: k-component lower bound}--\ref{thm: k-component general})}\label{sec: k-component}

\input{31_k_lower_bound}

\input{32_k_decreasing}

\input{33_k_full}

\input{34_k_general}


%% file: 31_k_lower_bound.tex
\subsection{Lower bound on optimal competitive ratio}\label{sec: k-component lower bound}


\begin{restateableTheorem}{thm: k-component lower bound}
  For $k$-Component Dynamization
  (and consequently for its generalizations)
  no deterministic online algorithm has ratio competitive ratio less than $k$.
\end{restateableTheorem}
  To develop intuition before we give the detailed proof for the general case, 
  here is a sketch of how the proof goes for $k=2$.
%
The adversary begins by inserting one item of weight 1 and one item of infinitesimal weight $\eps>0$, followed by a sequence of $n-2$ weight-zero items just until the algorithm's cover has just one component.
(This must happen, or the competitive ratio is unbounded --- \OPT pays only at time 1,
while the algorithm continues to pay at least $\eps$ each time step.)
By calculation the algorithm pays at least $2+(n-1)\eps$,
while \OPT pays $\min(2+\eps, 1+(n-1)\eps)$,
giving a ratio of $1.5-O(\eps)$.

This lower bound does not reach 2
(in contrast to the standard ``rent-or-buy'' lower bound)
because the algorithm and \OPT both pay a ``setup cost'' of 1 at time 1.
However, at the end of sequence, the algorithm and \OPT are left with a component of weight $\sim 1$ in place.
The adversary can now continue, doing a second phase without the setup cost,
by inserting an item of weight $\sqrt \eps$, then zeros
just until the algorithm's cover has just one component (again this must happen or the ratio is unbounded).
Let $m$ be the length of this second phase.
By calculation, for \emph{this} phase,
the algorithm pays at least $(m-1)\sqrt\eps+1$
while \OPT pays at most $\min(1+\sqrt{\eps}+\eps, (m-1)(\sqrt\eps+\eps))$,
giving a ratio of $2-O(\sqrt{\eps})$ for just the phase.

The ratio of the whole sequence (both phases together) is now $1.75 - O(\sqrt{\eps})$.
By doing additional phases (using infinitesimal $\eps^{1/i}$ in the $i$th phase),
the adversary can drive the ratio arbitrarily close to 2.
This is the idea for $k=2$.
Next we give the detailed proof for the general case (arbitrary $k\ge 2$).


\begin{labeledProof}[Proof of Theorem~\ref{thm: k-component lower bound}.]
  Fix an arbitrarily small $\eps>0$.
  Define $k+1$ sequences of items (weights) as follows.
  Sequence $\sigma{(k+1)}$ has just one item, $\sigma_{1}{(k+1)}=\eps$.
  For $j\in\{k,k-1,\ldots,1\}$, in decreasing order,
  define sequence $\sigma{(j)}$ to have
  $n_j=\lceil k/\sigma_1{(j+1)}\rceil$ items,
  with the $i$th item being $\sigma_i{(j)} = \eps^{n_{k}+n_{k-1}+\cdots+n_{j}-i+2}$.
  Each sequence $\sigma(j)$ is strictly increasing,
  and all items in $\sigma{(j)}$
  are smaller than all items in $\sigma{(j+1)}$.
  Every two items differ by a factor of at least $1/\eps$,
  so the cost to build any component will be at most $1/(1-\eps)$ times
  the largest item in the component.


  \paragraph{Adversarial input sequence $\inp$.} 
  Fix any deterministic online algorithm $\mysf A$.
  Define the input sequence $\inp$ to interleave the $k+1$ sequences in $\{\sigma{(j)} : 1\le j \le k+1\}$ as follows.
  Start by inserting the only item from sequence $\sigma(k+1)$: take $\inp_1 = \{\sigma_{1}{(k+1)}\}=\{\eps\}$.
  For each time $t\ge 1$,
  after $\mysf A$ responds to the insertion at time $t$, determine the next insertion $\inp_{t+1}=\{x\}$ as follows.
  For each sequence $\sigma{(j)}$, call the most recent (and largest) item inserted so far from $\sigma{(j)}$,
  if any, the \emph{representative} of the sequence.
  Define index $\ell(t)$ so that the largest representative in any new component at time $t$
  is the representative of $\sigma{(\ell(t))}$.
  (The item inserted at time $t$ is necessarily a representative
  and in at least one new component, so $\ell(t)$ is well-defined.)
  At time $t+1$ choose the inserted item $x$ to be the next unused item from sequence $\sigma{(\ell(t)-1)}$.
  %
  Define the \emph{parent} of $x$,
  %
  denoted $p(x)$, to be the representative of $\sigma{(\ell(t))}$ at time $t$.
  (Note: \mysf A's build cost at time $t$ was at least $p(x) \gg x$.)
  Stop when the cumulative cost paid by $\mysf A$ reaches $k$.
  This defines the input sequence $\inp$.

  \paragraph{The input $\inp$ is well-defined.} 
  Next we verify that $\inp$ is well-defined,
  that is, that \emph{(a)} $\ell(t)\ne 1$ for all $t$
  (so $x$'s specified sequence $\sigma(\ell(t)-1)$ exists)
  and \emph{(b)} each sequence $\sigma{(j)}$ is chosen at most $n_j$ times.
  First we verify \emph{(a)}.
  Choosing $x$ as described above forces the algorithm to maintain the following invariants
  at each time $t$:
  \begin{enumerate} \em
    \item[(i)] Each of the sequences
    in $\{\sigma(j) : \ell(t) \le j \le k+1\}$
    has a representative,
    and
  \item[(ii)] no two of these $k-\ell(t)+2$ representatives are in any one component.
  \end{enumerate}
  
  Indeed, the invariants hold at time $t=1$ when $\ell(t)=k+1$.
  Assume they hold at some time $t$.
  At time $t+1$ the newly inserted element $x$
  is the new representative of $\sigma(\ell(t)-1)$
  and is in some new component, so $\ell(t+1) \ge \ell(t)-1$.
  These facts imply that Invariant (i) is maintained.
  By the definition of $\ell(t+1)$,
  the components built at time $t+1$ contain the representative from $\sigma(\ell(t+1))$
  but no representative from any $\sigma(j)$ with $j > \ell(t+1)$.
  This and $\ell(t+1) \ge \ell(t)-1$ imply that Invariant (ii) is also maintained.

  By inspection, Invariants (i) and (ii) imply that
  \mysf A has at least $k-\ell(t)+2$ components at time $t$.
  But \mysf A has at most $k$ components, so $\ell(t)\ge 2$.

  Next we verify \emph{(b)}, that $\inp$ takes at most $n_j$ items from each sequence $\sigma{(j)}$.
  This holds for $\sigma(k+1)$
  just because, by definition, after time 1, $\inp$ cannot insert an item from $\sigma(k+1)$.
  Consider any $\sigma{(j)}$ with $j\le k$.
  For each item $\sigma_i(j)$ in $\sigma(j)$,
  when $\inp$ inserted $\sigma_i(j)$,
  algorithm \mysf A paid at least $p(\sigma_i(j)) \ge \sigma_1(j+1)$ at the previous time step.
  So, before all $n_j$ items from $\sigma(j)$ are inserted,
  \mysf A must pay at least $n_j\, \sigma_1(j+1) \ge k$
  (by the definition of $n_j$), and the input stops.
  It follows that $\inp$ is well-defined.

  \paragraph{Upper-bound on optimum cost.}
  Next we upper-bound the optimum cost for $\inp$.
  For each $j\in\{1,\ldots, k\}$, define  $\out{(j)}$ to be the solution for $\inp$
  that partitions the items inserted so far into the following $k$ components:
  one component containing items from $\sigma{(j)}$ and $\sigma{(j+1)}$,
  and, for each $h\in\{1,\ldots,k+1\}\setminus \{j,j+1\}$,  one containing items from $\sigma{(h)}$.

  To bound $\cost(\out{(j)})$, i.e., the total cost of new components in $\out(j)$,
  first consider the new components 
  such that the largest item in the new component is the just-inserted item, say, $x$.
  The cost of such a component is at most $x/(1-\eps)$.
  Each item $x$ is inserted at most once, so the total cost of all such
  components is at most $1/(1-\eps)$ times the sum of all defined items,
  and therefore at most $\sum_{i=1}^\infty \eps^i/(1-\eps) = \eps/(1-\eps)^2$.
  For every other new component, the just-inserted item $x$ must be from sequence $\sigma{(j+1)}$,
  so the largest item in the component is the parent $p(x)$ (in $\sigma{(j)}$)
  and the build cost is at most $p(x)/(1-\eps)$.
  Defining $m_j \le n_j$ to be the number of items inserted from $\sigma{(j)}$,
  the total cost of building all such components is
  at most $\sum_{i=1}^{m_j} p(\sigma_i{(j)})/(1-\eps)$.
  So $\cost(\out(j))$ is at most ${\eps}/{(1-\eps)^2} + \sum_{i=1}^{m_j} {p(\sigma_i{(j)})}/{(1-\eps)}$.

  The cost of $\OPT$ is at most $\min_j \cost(\out(j))$.
  The minimum is at most the average, so
  \begin{align*}
    (1-\eps)^2\cost(\OPT)
      \,\le\,
      \min_{j=1,\ldots,k} \eps + \sum_{i=1}^{m_j} p(\sigma_i{(j)})
      \,\leq\, 
      \eps + \frac{1}{k}\sum_{j=1}^k \sum_{i=1}^{m_j}p(\sigma_i{(j)}).
  \end{align*}
  \paragraph{Lower bound on algorithm cost.}
  The right-hand side of the above inequality is at most $(\eps/k +1/k) \cost(\mysf A)$,
  because
  $\cost(\mysf A) \ge k$ (by the stopping condition)
  and
  $\sum_{j=1}^k \sum_{i=1}^{m_j} p(\sigma_i{(j)})\le \cost(\mysf A)$.
  (Indeed, for each $j\in\{1,\ldots, k\}$ and $i\in\{1,\ldots, m_j\}$,
  the item $\sigma_{i}{(j)}$ was inserted at some time $t\ge 2$,
  and $\mysf A$ paid at least $p(\sigma_{i}{(j)})$ at the previous time $t-1$.)
  So the competitive ratio is at least $(1-\eps)^2/(\eps/k + 1/k) \ge (1-3\eps)k$.
  This holds for all $\eps>0$, so the ratio is at least $k$.
\end{labeledProof}


%% file: 32_k_decreasing.tex
\subsection{Upper bound for $k$-Component Dynamization with decreasing weights}\label{sec: k-component decreasing}

\begin{figure*}[t]
  \centering\framebox{\parbox{0.97\textwidth}{\sf\smaller[0]\setlength{\parskip}{1.5pt}

      \begin{steps}\small
      \item[]\underline{algorithm Greedy-Dual($\inp_1, \inp_2, \ldots, \inp_n$)}
        \mycomment{--- for $k$-Component Dynamization with decreasing weights}
        
        \step maintain a cover (collection of components), initially empty

        \step for each time $t=1,2,\ldots, n$ such that $\inp_t \ne \emptyset$:

        \begin{steps}

          \step if there are $k$ current components:

          \begin{steps}\label{block: if k}

            \step\label{line: increase credit}
            increase all components' credits continuously until some component $S$ has
            $\credit[S] \ge \wt_t(S)$

            \step  let $S_0$ be the oldest component such that $\credit[S_0] \ge \wt_t(S_0)$

            \step\label{line: merge}
            merge $I_t$, $S_0$ and all components newer than $S_0$ into one new component $S'$

            \step\label{line: initialize credit}
            initialize $\credit[S']$ to 0

          \end{steps}  

          \step\label{block: else} else:

          \begin{steps}
            \step create a new component from $\inp_t$, with zero credit
          \end{steps}
        \end{steps}
      \end{steps}
    }}
  \caption{A newest-first $k$-competitive algorithm for $k$-Component Dynamization with decreasing weights
    (Theorem~\ref{thm: k-component decreasing}).
  }~\label{fig: k-component decreasing}
\end{figure*}

Recall that, in $k$-Component Dynamization with decreasing weights,
each item $x\in\inp_t$ has weights $\wt_{t}(x) \ge \wt_{t+1}(x) \ge \cdots \ge \wt_n(x)$.
The cost of building a component $S\subseteq \univ t$ at time $t$ is redefined as $\wt_{t}(S) = \sum_{x\in S} \wt_t(x)$.
This variant is technically useful, as a stepping stone to the LSM variant.

\begin{restateableTheorem}{thm: k-component decreasing}
  For $k$-Component Dynamization with decreasing weights 
  (and plain $k$-Component Dynamization)
  the deterministic online algorithm
  in Figure~\ref{fig: k-component decreasing} is $k$-competitive.
\end{restateableTheorem}
\begin{proof}
  %
  %
  Consider any execution of the algorithm
  on any input $\inp_1, \inp_2, \ldots, \inp_n$.
  Let $\delta_t$ be such that each component's credit increases by $\delta_t$ at time $t$.
  (If Block~\ref{block: else} is executed, $\delta_t=0$.)
  To prove the theorem we show the following lemmas.

  \begin{lemma}\label{lemma: alg decreasing}
    The cost incurred by the algorithm is at most
    %
    $k \sum_{t=1}^n \big(\wt_t(\inp_t) + \delta_t\big)$.
    %
  \end{lemma}
  
\begin{lemma}\label{lemma: opt decreasing}
    The cost incurred by the optimal solution is at least
    %
    $\sum_{t=1}^n \big(\wt_t(\inp_t) + \delta_t\big)$.
    %
  \end{lemma}
  
  \begin{labeledProof}[Proof of Lemma~\ref{lemma: alg decreasing}.]
    As the algorithm executes, keep the components ordered by age, oldest first.
    Assign each component a \emph{rank} equal to its rank in this ordering.
    Say that the rank of any \emph{item} is the rank of its current component,
    or $k+1$ if the item is not yet in any component.
    At each time $t$,
    when a new component is created in Line~\ref{line: merge},
    the ranks of the items in $S_0$ stay the same,
    but the ranks of all other items decrease by at least 1.
    Divide the cost of the new component into two parts:
    the contribution from the items that decrease in rank,
    and the remaining cost.
    
    Throughout the execution of the algorithm,
    each item's rank can decrease at most $k$ times,
    so the total contribution from items as their ranks decrease
    is at most $k\sum_{t=1}^n \wt_t(\inp_t)$
    (using here that the weights are non-increasing with time).
    To complete the proof of the lemma,
    observe that the remaining cost
    is the sum, over times $t$
    when Line~\ref{line: merge} is executed,
    of the weight $\wt_t(S_0)$ of the component $S_0$ at time $t$.
    This sum is at most the total credit created,
    because, when a component $S_0$ is destroyed in Line~\ref{line: merge},
    at least the same amount of credit (on $S_0$) is also destroyed.
    But the total credit created is $k\sum_{t=1}^n \delta_t$,
    because when Line~\ref{line: increase credit} executes
    it increases the total component credit by $k\delta_t$.
  \end{labeledProof}

  
  \begin{labeledProof}[Proof of Lemma~\ref{lemma: opt decreasing}.]
    Let $\out^*$ be an optimal solution. 
    Let $\out$ denote the algorithm's solution.
    At each time $t$, when the algorithm executes Line~\ref{line: increase credit},
    it increases the credit of each of its $k$ components in $\out_{t-1}$ by $\delta_t$.
    So \emph{the total credit the algorithm gives is $k\sum_t \delta_t$.}

    For each component $S\in\out_{t-1}$, think of the credit given to $S$
    as being distributed over the component's items $x\in S$
    in proportion to their weights, $\wt_t(x)$:
    at time $t$, each item $x\in S$
    receives credit $\delta_t \wt_t(x)/\wt_t(S)$.
    Have each $x$, in turn,
    charge this amount to one component in \OPT's current cover $\out^*_t$ that contains $x$.
    In this way, \emph{the entire credit $k\sum_{t=1}^n \delta_t$ is charged
      to components in $\out^*$.}

      Recall that $\interval t {t'}$ denotes $\{t, t+1, \ldots, t'\}$.
    \begin{sublemma}\label{sublemma: item charge}
      Let $x$ be any item.
      Let $\interval t {t'}$ be any time interval throughout which $x$ remains in the same component in $\out$.
      The cumulative credit given to $x$ during $\interval {t} {t'}$ is at most $\wt_t(x)$.
    \end{sublemma}
    \begin{proof}
      Let $S$ be the component in $\out$ that contains $x$ throughout $\interval t {t'}$.
      Assume that $\delta_{t'} > 0$ (otherwise reduce $t'$ by one).
      Let $\credit_{t'}[S]$ denote $\credit[S]$ at the end of iteration $t'$.
      Weights are non-increasing with time,
      so the credit that $x$ receives during $\interval t {t'}$ is
      \begin{align*}
        \sum_{i=t}^{t'} \frac{\wt_i(x)}{\wt_i(S)} \delta_i
        \,\le\,
        \frac{\wt_t(x)}{\wt_{t'}(S)} \sum_{i=t}^{t'} \delta_i
        \,\le\,
        \frac{\wt_t(x)}{\wt_{t'}(S)} \credit_{t'}[S].
      \end{align*}
      %
      The right-hand side is at most $\wt_t(x)$.
      %
        (Indeed, 
        in iteration $t'$
        Line~\ref{line: increase credit} increased all components' credits by $\delta_{t'}>0$,
        while maintaining the invariant that $\credit[S] \le \wt_{t'}(S)$,
        so $\credit_{t'}[S] \le \wt_{t'}(S)$.)
    \end{proof}
    Next we bound how much charge \OPT's components (in $\out^*$) receive.
    For any time $t$,
    let $\calN_t^*=\out_t^*\setminus\out_{t-1}^*$
    contain the components that \OPT creates at time $t$,
    and let $N_t^* = \bigcup_{S \in \calN^*_t} S$
    contain the items in these components.
    Call the charges received by components in $\calN^*_t$
    from components created by the algorithm before time $t$
    \emph{forward charges}.
    Call the remaining charges (from components created by the algorithm at time $t$ or after)
    \emph{backward charges}.

    Consider first the \emph{backward} charges to components in $\calN^*_t$.
    These charges come from components in $\out_{t-1}$,
    via items $x$ in $N_t^*\cap \univ {t-1}$,
    from time $t$ until the algorithm destroys the component in $\out_{t-1}$ that contains $x$.
    By Sublemma~\ref{sublemma: item charge},
    the total charge via a given $x$ from time $t$ until its component is destroyed is at most $\wt_t(x)$,
    so the cumulative charge to components in $\calN^*_t$
    from older components
    is at most
    %
      \(
      \wt_t(N_t^*\cap\univ {t-1})
      = \wt_t(N_t^*) - \wt_t(\inp_t) 
      \)
    (using that $N_t^*\setminus \univ {t-1} = \inp_t$).
    Using that \OPT pays at least $\wt_t(N_t^*)$ at time $t$,
    and summing over $t$, 
    \emph{the sum of all backward charges is at most} $\cost(\OPT) - \sum_t \wt_t(\inp_t)$.

    Next consider the \emph{forward} charges, from components created at time $t$ or later,
    to any component $S^*$ in $\calN^*_t$.
    Component $S^*$ receives no forward charges at time $t$,
    because components created by the algorithm at time $t$ receive no credit at time $t$.
    Consider the forward charges $S^*$ receives at any time $t' \ge t+1$.
    At most one component (in $\out_{t'-1}$) can contain items in $N_t^*$,
    namely, the component in $\out_{t'-1}$ that contains $\inp_t$.
    (Indeed, the algorithm merges components ``newest first'',
    so any other component in $\out_{t'-1}$ created after time $t$ only contains items inserted after time $t$,
    none of which are in $N_t^*$.)
    At time $t'$, the credit given to that component is $\delta_{t'}$,
    so the components created by the algorithm at time $t'$
    charge a total of at most $\delta_{t'}$ to $S^*$.
    Let $m(t, t') = |\calN^*_t \cap \out^*_{t'}|$ be the number of components $S^*$
    that \OPT created at time $t$ that remain at time $t'$.
    Summing over $t' \ge t+1$ and $S^*\in\calN^*_t$,
    the forward charges to components in $\calN^*_t$ total at most
    \(\sum_{t'=t+1}^n m(t, t') \delta_{t'}\).
    Summing over $t$, \emph{the sum of all forward charges is at most}
    \begin{align*}
      \sum_{t=1}^n \sum_{t'=t+1}^n m(t, t') \delta_{t'}
        ~=~
      \sum_{t'=2}^n \delta_{t'} \sum_{t=1}^{t'-1} m(t, t')  
          ~\le~
      \sum_{t'=1}^n \delta_{t'} (k-1)
    \end{align*}
    (using that $\sum_{t=1}^{t'-1} m(t, t') \le k-1$ for all $t$,
    because \OPT has at most $k$ components at time $t'$,
    at least one of which is created at time $t'$).

    Recall that the entire credit $k\sum_{t=1}^n \delta_t$ is charged to components in $\out^*$.
    Summing the bounds from the two previous paragraphs
    on the (forward and backward) charges,
    this implies that
    \begin{align*}
      \textstyle 
      k\sum_{t=1}^n \delta_t
      \le
      \cost(\OPT) - \sum_{t=1}^n \wt_t(\inp_t) + (k-1)\sum_{t=1}^n \delta_t.
    \end{align*}
    This proves the lemma, as it is equivalent to the desired bound
    $\cost(\OPT) \ge \sum_{t=1}^n \wt_t(\inp_t) + \delta_t$.
  \end{labeledProof}

  This proves Theorem~\ref{thm: k-component decreasing}.
\end{proof}


%% file: 33_k_full.tex
\subsection{Bootstrapping newest-first algorithms}\label{sec: k-component LSM}

\begin{restateableTheorem}{thm: meta}
  Any newest-first online algorithm $\mysf A$ for $k$-Component (or Min-Sum) dynamization
  with decreasing weights
  can be converted into an equally competitive online algorithm $\mysf A'$ for the LSM variant.
\end{restateableTheorem}
\begin{proof}
  \newcommand{\nr}{\mymathname{nr}}
  
  Fix an instance $(\inp, \wt)$ of LSM $k$-Component (or Min-Sum) Dynamization.
  For any solution $\out$ to this instance, let $\wt(\out)$ denote its build cost
  using build-cost function $\wt$.
  For any set $S$ of items and any item $x\in S$,
  let $\nr(x, S)$ be 0 if $x$ is redundant in $S$
  (that is, there exists a newer item in $S$ with the same key)
  and 1 otherwise.
  Then $\wt_t(S) = \sum_{x\in S} \nr(x, S) \wt_t(x)$,
  where $\wt_t(x)$ is $\wt(x)$ unless $x$ is expired, in which case $\wt_t(x)$ is the tombstone weight of $x$.
  The tombstone weight of $x$ must be at most $\wt(x)$,
  so $\wt_t(x)$ is non-increasing with $t$.

  For any time $t$ and item $x\in\univ t$, define $\wt'_t(x) = \nr(x, \univ t) \wt_t(x)$.
  For any item $x$, $\wt'_t(x)$ is non-increasing with $t$,
  so $(\inp, \wt')$ is an instance of $k$-Component Dynamization with decreasing weights.
  For any solution $\out$ for this instance, let $\wt'(\out)$ denote its build cost
  using build-cost function $\wt'$.
  
  \begin{lemma}\label{lemma: redundant 1}
    For any time $t$ and set $S\subseteq\univ t$,
    we have $\wt'_t(S) \le \wt_t(S)$.
  \end{lemma}
  \begin{proof}
    Redundant items in $S$ are redundant in $\univ t$, so
    \begin{align}
      \wt'_t(S)
      = \sum_{x\in S} \wt'_t(x)
        = \sum_{x\in S} \nr(x, \univ t) \wt_t(x)
        \le \sum_{x\in S} \nr(x, S) \wt_t(x)
      = \wt_t(S). \label{eqn: redundant}
    \end{align}
  \end{proof}
  
  \begin{lemma}\label{lemma: redundant 2}
    Let $\out$ be any newest-first solution for $(\inp, \wt')$ and $(\inp, \wt)$.
    Then $\wt'(\out) = \wt(\out)$.
  \end{lemma}
  \begin{proof}
    Consider any time $t$ with $\inp_t \ne \emptyset$.
    Let $S$ be $\out$'s new component at time $t$ (so $\out_t\setminus\out_{t-1} = \{S\}$).
    Consider any item $x\in S$.
    Because $\out$ is newest-first, $S$ includes all items inserted with or after $x$.
    So $x$ is redundant in $\univ t$ iff $x$ is redundant in $S$,
    that is, $\nr(x, \univ t) = \nr(x, S)$,
    so $\wt'_t(S) = \wt_t(S)$
    (because Bound~\eqref{eqn: redundant} above holds with equality).
    Summing over all $t$ gives $\wt'(\out) = \wt(\out)$.
  \end{proof}

  Given an instance $(I, \wt)$ of LSM $k$-Component Dynamization, 
  %
  the algorithm $\mysf A'$ simulates $\mysf A$
  %
  on the instance $(I, \wt')$ defined above.
  Using Lemma~\ref{lemma: redundant 2},
  that $\mysf A$ is $c$-competitive,
  and $\wt'(\OPT(I, \wt')) \le \wt(\OPT(I, \wt))$
  (by Lemma~\ref{lemma: redundant 1}), we get
  \begin{align*}
    \wt(\mysf A'(I, \wt)) 
      =
      \wt'(\mysf A(I, \wt'))
      \le
    c \wt'(\OPT(I, \wt'))
    \le
    c \wt(\OPT(I, \wt)).
  \end{align*}
  %
  So $\mysf A'$ is $c$-competitive.
\end{proof}


Combined with the observation that the Greedy-Dual algorithm (Figure~\ref{fig: k-component decreasing}) is newest-first, Theorems~\ref{thm: k-component decreasing} and~\ref{thm: meta}
yield a $k$-competitive algorithm for LSM $k$-Component Dynamization:

\begin{restateableCorollary}{cor: k-component LSM}
  There is a deterministic online algorithm for LSM $k$-Component Dynamization
  with competitive ratio $k$.
\end{restateableCorollary}


%% file: 34_k_general.tex
\subsection{Upper bound for general variant}\label{sec: k-component general}

\begin{restateableTheorem}{thm: k-component general}
  For general $k$-Component Dynamization,
  the deterministic online algorithm $\mysf B_k$ in Figure~\ref{fig: k-component general}
  is $k$-competitive.
\end{restateableTheorem}

\begin{figure*}[t!]
  \!\!\framebox{\parbox{0.98\textwidth}{\sf\smaller[0]\setlength{\parskip}{1.5pt}

      \begin{steps}
        
      \item[] \underline{algorithm $\mysf B_1(\inp_1, \inp_2, \ldots, \inp_n)$}  \mycomment{--- for $k=1$}

        \step for $t=1,2,\ldots, n$: use cover $\out_t= \{\univ t\}$ where $\univ t = \bigcup_{i=1}^t \inp_i$
        \mycomment{---  one component holding all items}
      \end{steps}
      \smallskip 
      \begin{steps}
        
      \item[] \underline{algorithm $\mysf B_k(\inp_1, \inp_2, \ldots, \inp_n)$}  \mycomment{--- for $k\ge 2$}
        
        \step initialize $t'=1$
        \mycomment{--- variable $t'$ holds the start time of the current phase}

        \step for $t=1,2,\ldots, n$:

        \begin{steps}

          \step let $\out' = \mysf B_{k-1}(\inp_{t'}, \inp_{t'+1}, \ldots, \inp_t)$
          \mycomment{--- the solution generated by $\mysf B_{k-1}$ for the current phase so far}
          
          \step if the total cost of $\out'$ exceeds $(k-1)\wt_t(\univ t)$:
          take $\out_t = \{\univ t\}$ and let $t'=t+1$
          \mycomment{---  end the phase}
          
          \step else: use cover $\out_t = \{\univ {t'}\} \cup \out'_t$,
          where $\out'_t$ is the last cover in $\out'$
          \mycomment{--- $\out'_t$ has at most $k-1$ components}
        \end{steps}
      \end{steps}
    }}
  \caption{Recursive algorithm for general $k$-Component Dynamization
    (Theorem~\ref{thm: k-component general}).}\label{fig: k-component general}
\end{figure*}

\begin{proof}
  The proof is by induction on $k$.
  For $k=1$, Algorithm $\mysf B_1$ is 1-competitive (optimal) because there is only one solution for any instance.
  Consider any $k\ge 2$, and assume inductively that $\mysf B_{k-1}$ is $(k-1)$-competitive.
  Fix any input $(\inp, w)$ with $\inp=(\inp_1, \ldots, \inp_n)$.
  Let $\OPT_k$ denote the optimal (offline) algorithm,
  and let $\out^* = \OPT_k(\inp_1, \ldots, \inp_n)$ be an optimal solution for $\inp$.

  Let $\mathcal N^*_t = \out^*_t \setminus \out^*_{t-1}$ denote \OPT's new components at time $t$.
  %
    For $a, b\in \mathbb [n]$,
    let $\Delta_a^b(\OPT_k)$
  denote the cost incurred by $\OPT_k$ during time interval $\interval a b$,
  that is, $\sum_{i=a}^{b} \sum_{S\in \mathcal N^*_i} w_i(S)$.
  (Recall $\interval a b$ denotes $\{a, a+1, \ldots, b\}$.)
  Likewise,
  %
  let $\Delta_a^b(\mysf B_k)$
  denote the cost incurred by $\mysf B_k$ during $[a, b]$.
  Let $\inputs a {b}=(\inp_a, \inp_{
    a+1}, \ldots, \inp_{b})$ denote the subproblem
  formed by the insertions during $[a,b]$,
  with build-costs inherited from $w$.

  Recall that $\mysf B_k$ partitions the input sequence into \emph{phases},
  each of which (except possibly the last) ends with $\mysf B_k$ doing a \emph{full merge}
  (i.e., at a time $t$ with $|\out_t|=1$).
  Assume without loss of generality that $\mysf B_k$ ends the last phase with a full merge.
  (Otherwise, append a final empty insertion at time $n+1$ and define $w_{n+1}(\univ {n+1}) = 0$.
  This does not increase the optimal cost, and causes the algorithm to do a full merge at time $n+1$
  unless its total cost in the phase is zero.)
  Consider any phase.
  Now fix $a$ and $b$ to be the first and last time steps during the phase.
  To prove the theorem, we show $\Delta_a^{b}(\mysf B_k) \le k\, \Delta_a^{b}(\OPT_k)$.
  The theorem follows by summing over the phases.
  
  The proof is via a series of lemmas.
  Recall that $\univ t$ denotes $\bigcup_{i=1}^t \inp_i$. 

  \begin{lemma}\label{lemma: induction}
    %
    For any integer $j\in [a,b]$,~\(
    \cost(\mysf B_{k-1}(\inputs a {j}))
    \,\le\,
    (k-1)  \cost(\OPT_{k-1}(\inputs a {j})).
    \)
  \end{lemma}
  \begin{proof}
      The instance $(\inp, w)$ obeys Restrictions \textbf{(R1)}--\textbf{(R3)}.
      So, by inspection of those restrictions, $\inputs a {j}$ also obeys them.
      That is, $\inputs a {j}$ is a valid instance of general $(k-1)$-component Dynamization.
      So, by the inductive assumption, $\mysf B_{k-1}$ is $(k-1)$-competitive for $\inputs a j$. 
  \end{proof}


  For $j\in [a, b]$, say that \emph{\OPT rebuilds by time $j $}
  if
  %
  $\univ {a-1} \subseteq \bigcup_{i=a}^j \bigcup_{S\in \mathcal N_i} S$.
  %
  That is, every element inserted before time $a$ is in some new component during $[a, j]$.
  %
  (Equivalently, $\bigcup_{i=a}^j \bigcup_{S\in \mathcal N_i} S = \univ j$.)
  %
  \begin{lemma}\label{lemma: new cover}
    Suppose \OPT rebuilds by time $j$.
    Then $\Delta_a^j(\OPT_k) \ge w_{j}(\univ j)$.
  \end{lemma}
  \begin{proof}
    \begin{align*}
      w_{j}(\univ j)
      & \textstyle {} = w_{j}\big(\bigcup_{i=a}^j\bigcup_{S\in \mathcal N_i}S \big)
      && (\textit{\OPT rebuilds by time } j) \\
      & \textstyle {} \le  \sum_{i=a}^j \sum_{S\in \mathcal N_i} w_j(S)
      && (\textit{by sub-additivity }\textbf{(R1)}) \\
      & \textstyle {} \le  \sum_{i=a}^j \sum_{S\in \mathcal N_i} w_i(S)
      && (\textit{by temporal monotonicity } \textbf{(R3)}) \\
      & {} = \Delta_a^j(\OPT_k).
      && (\textit{by definition})
    \end{align*}
  \end{proof}

  \begin{lemma}\label{lemma: not new}
    Suppose \OPT does not rebuild by time $j\in [a, b]$.
    Then 
    %
    $\cost(\OPT_{k-1}(\inputs a j)) \le \Delta_a^j(\OPT_k)$.
    %
  \end{lemma}
  \begin{proof}
    Because \OPT does not rebuild by time $j$,
    some element $x$ in $\univ {a-1}$ is not in any new component during $[a, j]$.
    %
    Let $S$ be the component in $\out^*_j$ containing $x$.
    Since $S^*$ is not new during $[a, j]$,
    it must be that $S^*$ is in $\out^*_i$ for every $i\in [a-1, j]$,
    and $S^*\subseteq \univ {a-1}$.

    %
      For the subproblem $\inputs a j$,
      let $\out'$ be the solution
      defined by
      $\out_i' = \{S\setminus \univ {a-1} : S \in \out^*_i\}\setminus\{\emptyset\}$
      for $i\in [a,j]$.
    %
    Because each $\out^*_i$ has at most $k$ components,
    one of which is  $S^*$, and $S^*\subseteq \univ {a-1}$,
    it follows that each $\out_i$ has at most $k-1$ components.
    So $\cost(\OPT_{k-1}(\inputs a j)) \le \cost(\out')$.
    
    If a given component $S \setminus \univ {a-1}$ is new in $\out'$ at time $i\in [a, j]$,
    then the corresponding component $S$ is new in $\out^*$ at time $i$.
    Further,
    %
    by suffix monotonicity \textbf{(R2)},
    %
    the cost $w_i\big(S \setminus\univ {a-1}\big)$ paid by $\out'$ for $S\setminus\univ {a-1}$
    is at most the cost $w_i(S)$ paid by $\out^*$ for $S$.
    (Inspecting the definition of \textbf{(R2)},
    we require that $S\ne \univ i$, which holds because \OPT has not rebuilt by time $j$.)
    So $\cost(\out') \le \Delta_a^j(\OPT_k)$.  
  \end{proof}

  \begin{lemma}\label{lemma: no merge}
    $\cost(\mysf B_{k-1}( \inputs a {b-1})) \le (k-1)\Delta_a^{b-1}(\OPT_k)$
  \end{lemma}
  \begin{proof}
    If $a=b$ then $\cost(\mysf B_{k-1}( \inputs a {b-1})) = 0$, so assume $a<b$.
    If \OPT rebuilds by time $b-1$, then
    \begin{align*}
      \cost(\mysf B_{k-1}( \inputs a {b-1})) 
      & {} \le (k-1) w_{b-1}(\univ {b-1})
      && (\textit{\(B_k\) does not end the phase at time } b-1) \\
      & {} \le (k-1) \Delta_a^{b-1}(\OPT_k)
      && (\textit{Lemma~\ref{lemma: new cover} with } j=b-1).
      \\
      \intertext{Otherwise \OPT does not rebuild by time $b-1$, so}
      \cost(\mysf B_{k-1}( \inputs a {b-1}))
      & {} \le (k-1)\cost(\OPT_{k-1}( \inputs a {b-1}))
      && (\textit{Lemma~\ref{lemma: induction} with } j=b-1) \\
      & {} \le (k-1)\Delta_a^{b-1}(\OPT_k)
      && (\textit{Lemma~\ref{lemma: not new} with } j=b-1). \\[-2em]
      \rule{1.1in}{0pt}&\rule{2.1in}{0pt}&& \rule{2.8in}{0pt}
    \end{align*}
  \end{proof}
  
  \begin{lemma}\label{lemma: merge}
     $w_b(\univ b) \le \Delta_a^b(\OPT_k)$
  \end{lemma}
  \begin{proof}
    If \OPT rebuilds by time $b$,
    then
    \begin{align*}
      w_{b}(\univ {b})
      & {} \le \Delta_a^b(\OPT_k)
      && (\textit{Lemma~\ref{lemma: new cover} with \(j=b\))}.
      \\
      \intertext{Otherwise \OPT does not rebuild by time $b$, so}
      w_{b}(\univ {b})
      & {} < \cost(\mysf B_{k-1} (\inputs a b))/(k-1)
      && (\textit{because \(B_k\) ends the phase at time } b) \\
      & {} \le \cost(\OPT_{k-1} (\inputs a b))
      && (\textit{Lemma~\ref{lemma: induction} with } j=b) \\
      & {} \le \Delta_a^b(\OPT_k)
      && (\textit{Lemma~\ref{lemma: not new} with } j=b).\\[-2em]
      \rule{1.1in}{0pt}&\rule{2.1in}{0pt}&& \rule{2.8in}{0pt}
    \end{align*}
  \end{proof}

  The lemmas imply that the algorithm is $k$-competitive for the phase:
  \begin{align*}
    \Delta_a^b(\mysf B_k)
    &{} = ~\cost(\mysf B_{k-1}( \inputs a {b-1})) \,+ 
     w_{b}(\univ {b})
    &&(\textit{by definition of } \mysf B_k)
    \\
    & {} \le ~\cost(\mysf B_{k-1}( \inputs a {b-1})) \, + 
    \Delta_a^b(\OPT_k)
    && (\textit{Lemma~\ref{lemma: merge}})
    \\
    &{} \le (k-1)\Delta_a^{b-1}( \OPT_k)+\,
    \Delta_a^b(\OPT_k)
    &&(\textit{Lemma~\ref{lemma: no merge}})
    \\
    &{} \le k\, (\Delta_a^b \OPT_k)
    &&(\textit{as } \Delta_a^{b-1} (\OPT_k)  \le \Delta_a^b(\OPT_k) )
  \end{align*}
  Theorem~\ref{thm: k-component general} follows by summing over the phases.
\end{proof}


%% file: 40_offline.tex
\section{Properties of optimal offline solutions}\label{sec: offline}

\begin{restateableTheorem}{thm: offline}
  Every instance of $k$-Component or Min-Sum Dynamization
  has an optimal solution that is newest-first and lightest-first.
\end{restateableTheorem}


  \begin{figure*}
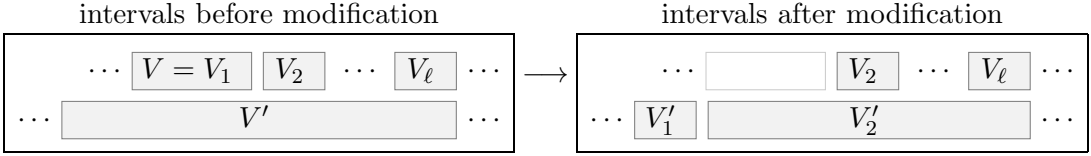

    \begin{displaymath}
      \begin{array}{|r|}
        \multicolumn{1}{c}{\text{intervals before modification}} \\ \hline
        \rule{0ex}{3.5ex}
        \cdots \iBox{1.8}{V=V_1} \iBox{1}{V_2} \iFix{0.8}{\cdots\!\!} \iBox{1}{V_\ell} \cdots
        \\ \cdots \iBox{5.6}{V'} \cdots \\[5pt] \hline
      \end{array}
      \longrightarrow
      \begin{array}{|r|}
        \multicolumn{1}{c}{\text{intervals after modification}} \\ \hline
        \rule{0ex}{3.5ex}
        \cdots \iGhost{1.8}{} \iBox{1}{V_2} \iFix{0.8}{\cdots\!\!} \iBox{1}{V_\ell} \cdots
        \\ \cdots \iBox{1}{V_1'}\iBox{4.6}{V_2'}\cdots \\[5pt] \hline
      \end{array}
    \end{displaymath}
    \caption{Replacing intervals $V$ and $V'$ by $V'_1$ and $V'_2$ (proof of Theorem~\ref{thm: offline}).}\label{fig: lightest-first}
  \end{figure*}


\begin{proof}
  Fix an instance $\inp=(\inp_1, \ldots, \inp_n)$.
  %
  Recall that $\interval t {t'}$ denotes $\{t,t+1,\ldots, t'\}$.
  For any component $S$ that is new at some time $t$ of a given solution $\out$,
  we say that $S$
  \emph{uses (time) interval $\interval t {t'}$},
  where $t' = \max\{j \in \interval t n : (\forall i\in \interval t j)~S\in \out_i\}$
  is the time that (this occurrence of) $S$ is destroyed.
  We refer to $[t,t']$ as the \emph{interval} of (this occurrence of) $S$.
  For the proof we think of any solution $\out$ as being constructed in two steps:
  (i) choose the set $T$ of time intervals 
  that the components of $\out$ will use, then
  (ii) given $T$, for each interval $\interval t {t'}\in T$,
  choose a set $S$ of items for $\interval {t} {t'}$, then form a component $S$ in $\out$
  with interval $\interval t {t'}$
  (that is, add $S$ to $\out_i$ for $i\in \interval t {t'}$).
  We shall see that the second step (ii) decomposes by item:
  an optimal solution can be found by greedily choosing the intervals for each item $x\in \univ n$ independently.
  The resulting solution has the desired properties.
  Here are the details.

  Fix an optimal solution $\out^*$ for the given instance,
  breaking ties by choosing $\out^*$ to minimize the total query cost
  $\sum_{\interval t {t'} \in T^*} t' - t + 1$
  where $T^*$ is the set of intervals of components in $\out^*$.
  Assume without loss of generality that, for each $t\in \interval 1 n$, if $\inp_t=\emptyset$,
  then $\out^*_t = \out^*_{t-1}$ (interpreting $\out^*_0$ as $\emptyset$).
  (If not, replace $\out^*_{t}$ by $\out^*_{t-1}$.)
  For each item $x\in\univ n$, let
  $\alpha^*(x)$ denote the set of intervals in $T^*$ of components that contain $x$.
  The build cost of $\out^*$
  equals $\sum_{x\in\univ n} \wt(x)\, |\alpha^*(x)|$.
  For each time $t$ and item $x\in\inp_t$, the intervals $\alpha^*(x)$ of $x$
  %
  \emph{cover} $\interval t n$,
  %
  meaning that the union of the intervals in $\alpha^*(x)$ is $\interval t n$.

  Next construct the desired solution $\out'$ from $T^*$.
  For each time $t$ and item $x\in \inp_t$,
  let $\alpha(x) = \{ V_1, \ldots, V_\ell \}$ be a sequence
  of intervals chosen greedily from $T^*$ as follows.
  Interval $V_1$ is the latest-ending interval starting at time $t$.
  For $i\ge 2$, interval $V_i$ is the latest-ending interval starting at time $t'_{i-1}+1$ or earlier,
  where $t'_{i-1}$ is the end-time of $V_{i-1}$.
  The final interval has end-time $t'_{\ell} = n$.
  By a standard argument, this greedy algorithm chooses from $T^*$ a minimum-size 
  interval cover of $\interval t n$, so $|\alpha(x)| \le |\alpha^*(x)|$.

  Obtain $\out'$  as follows:
  for each interval $\interval i j \in T^*$,
  add a component in $\out'$ with time interval $\interval i j$
  containing the items $x$ such that $\interval i {j} \in \alpha(x)$.
  This is a valid solution because, for each time $t$ and $x\in\inp_t$,
  $\alpha(x)$ covers $\interval t n$.
  Its build cost is at most the build cost of $\out^*$,
  because $\sum_{x\in\univ n} \wt(x) |\alpha(x)| \le \sum_{x\in\univ n} \wt(x) |\alpha^*(x)|$.
  At each time $t$,
  its query cost is at most the query cost of $\out^*$,
  because it uses the same set $T^*$ of intervals.
  So $\out'$ is an optimal solution.
  
  \paragraph{$\out'$ is newest-first.}
  The following properties hold:

  \begin{enumerate}
  \item
    \emph{$\alpha$ uses (assigns at least one item to) each interval $V\in T^*$.}
    Otherwise removing $V$ from $T^*$ (and using the same $\alpha$) would give a solution with the same build cost
    but lower query cost, contradicting the definition of $\out^*$.

  \item
    \emph{For all $t\in \interval 1 n$, 
      the number of intervals in $T^*$ starting at time $t$ is 1 if $\inp_t\ne\emptyset$ and 0 otherwise.}
    Among intervals in $T^*$ that start at $t$,
    only one --- the latest ending --- can be used in any $\alpha(x)$.
    So by Property 1 above, $T^*$ has at most interval starting at $t$.
    If $\inp_t\ne\emptyset$, $\out^*$ must have a new component at time $t$,
    so there is such an interval.
    If $\inp_t=\emptyset$ there is not
    (by the initial choice of $\out^*$ it has no new component at time $t$).

  \item
    \emph{
      For every two consecutive intervals $V_i, V_{i+1}$ in any $\alpha(x)$,
      $V_{i+1}$ is the interval in $T^*$ that starts just after $V_i$ ends.
    }
    Fix any such $V_i, V_{i+1}$.
    For every other item $y$ with $V_i\in \alpha(y)$,
    the interval following $V_i$ in $\alpha(y)$ must also (by the greedy choice) be $V_{i+1}$.
    That is, \emph{every} item assigned to $V_i$ is also assigned to $V_{i+1}$.
    If $V_{i+1}$ were to overlap $V_i$,
    replacing $V_i$ by the interval $V_i \setminus V_{i+1}$
    (within $T^*$ and every $\alpha(x)$)
    would give a valid solution with the same build cost but smaller total query cost, 
    contradicting the choice of $\out^*$.
    So $V_{i+1}$ starts just after $V_i$ ends.
    By Property 2 above, $V_{i+1}$ is the only interval starting then.

  \item
    \emph{
      For every pair of intervals $V$ and $V'$ in $T^*$,
      either $V\cap V'=\emptyset$, or one contains the other.
    }
    Assume otherwise for contradiction, that is, two intervals \emph{cross}:
    $V\cap V' \ne \emptyset$ and neither contains the other.
    Let $\interval a {a'}$ and $\interval b {b'}$ be a rightmost crossing pair in $T^*$,
    that is, such that $a < b < a' < b'$
    and no crossing pair lies in $\interval {a+1} n$.
    By Property 1 above, $\interval a {a'}$ is in some $\alpha(x)$.
    Also $a' < n$.
    Let $\interval {a'+1} c$ be the interval added greedily to $\alpha(x)$ following $\interval a {a'}$.
    (It starts at time $a'+1$ by Property 3 above.)
    The start-time of $ \interval b {b'}$ is in $ \interval a {a'+1}$ (as $a < b < a'$),
    so by the greedy choice (for $\interval a {a'}$)
    $\interval b {b'}$ ends no later than $\interval {a'+1} c$.
    Further, by the tie-breaking in the greedy choice, $c>b'$.
    So $ \interval {a'+1} c$ crosses $ \interval b {b'}$,
    contradicting that no crossing pair lies in $ \interval {a+1} n$.
  \end{enumerate}

  By inspection of the definition of newest-first,
  Properties 2 and 4 imply that $\out'$ is newest-first.

  \paragraph{$\out'$ is lightest-first.}
  To finish we show that $\out'$ is lightest-first.
  For any time $t\in \interval 1 n$, consider any intervals $V, V'\in T^*$ where $V$ ends at time $t$
  while $V'$ includes $t$ but does not end then.  To prove that $\out'$ is lightest-first,
  we show $\wt(V) < \wt(V')$.

  The intervals of $\out'$ are nested (Property 4 above),
  so $V \subset V'$ and the items assigned to $V=V_1$ are subsequently assigned (by Property 3 above)
  to intervals $V_2, \ldots, V_\ell$ within $V'$ as shown in Figure~\ref{fig: lightest-first},
  with $V_\ell$ and $V'$ ending at the same time.
  Since $V'$ does not end when $V$ does, $\ell\ge 2$.
  Consider modifying the solution $\out'$ as follows:
  remove intervals $V$ and $V'$ from $T^*$,
  and replace them by intervals $V_1'$ and $V_2'$
  obtained by splitting $V'$ so that $V'_2$ starts when $V$ started.
  (See the right side of Figure~\ref{fig: lightest-first}.)

  Reassign all of $V'$'s items to $V'_1$ and $V'_2$.
  Reassign all of $V$'s items to $V'_2$
  and unassign those items from each interval $V_i$.
  This gives another valid solution.
  It has lower query cost (as $V$ is gone),
  so by the choice of $\out^*$ (including the tie-breaking) the new solution must have strictly larger build cost.
  That is, the change in the build cost,
  $\wt(V)(1-\ell) + \wt(V')$, must be positive,
  implying that $\wt(V') > \wt(V)(\ell-1) \ge \wt(V)$ (using $\ell\ge 2$).
  Hence $\wt(V') > \wt(V)$.
\end{proof}


%% file: 50_conclusion.tex
\section{Concluding remarks}\label{sec: conclusion}

  This paper brings competitive analysis to bear on data-structure dynamization for non-uniform inputs,
  via two new online covering
  problems---Min-Sum Dynamization and $k$-Component Dynamization---for
  which it gives deterministic online algorithms
  with competitive ratios $\Theta(\log^* m)$ and $k$, respectively.
  The algorithms extend to handle lazy updates and deletions
  as they occur in industrial LSM systems.

  The paper also shows the existence of optimal offline solutions
  that are newest-first and lightest-first.
  As mentioned in the introduction,
  one consequence is that Bentley and Saxe's transforms
  give optimal solutions (up to lower-order terms)  for uniform inputs.
  Another is that, for Min-Sum and $k$-Component Dynamization,
  optimal solutions can be computed
  in time $O(n^3)$ and $O(k n^3)$, respectively,
  because optimal newest-first solutions can be computed in these time bounds
  via natural dynamic programs.

\subsection{Open problems}
Here are a few of many interesting problems that remain open.
\noindent  
For $k$-Component Dynamization:
\begin{itemize}[label=--]
  \item
    Is there an online algorithm with competitive ratio $O(\min(k, \log^* m))$?
  \item
    Is there an algorithm with competitive ratio $O(k/(k-h+1))$ versus $\OPT_h$ (the optimal solution with maximum query cost $h\le k$)?
  \item
    Is there a randomized algorithm with competitive ratio $o(k)$? 
\item
   A \emph{memoryless} randomized algorithm with competitive $O(k)$?
 \end{itemize}

\smallskip \noindent 
For Min-Sum Dynamization:

\begin{itemize}[label=--]
\item
  Is there an $O(1)$-competitive algorithm?

\item
Some LSM architectures only support newest-first algorithms.
Is there a newest-first algorithm with competitive ratio $O(\log^* m)$?

\item
What are the best ratios for the LSM and general variants? 

\end{itemize}

\smallskip \noindent
For both problems:

\begin{itemize}[label=--]
\item
  For instances $\inp$ that occur in practice,
  the ratio $\max_{t, t'} \wt(\inp_t)/\wt(\inp_{t'})$
  (for $t'$ such that $\wt(\inp_{t'}) > 0$) is often bounded.
  Does restricting to such instances allow smaller competitive ratios? 
\item
  For the decreasing-weights and LSM variants,
  is there always an optimal newest-first solution?
\end{itemize}

\subsection{Variations on the model}\label{sec: variations}

\paragraph{Tombstones deleted during major compactions.}
Times when the cover $\out_t$ has just one component (containing all inserted items)
are called \emph{full merges} or~\emph{major compactions}.
At these times, LSM systems delete all tombstone items (even non-redundant tombstones).
Our definition of LSM $k$-Component Dynamization does not capture this,
but our definition of General $k$-Component Dynamization does,
so the algorithm $\mysf B_k$ in Figure~\ref{fig: k-component general}
is $k$-competitive in this case.

\paragraph{Monolithic builds.}  
Our model underestimates query costs
because it assumes that new components can be built in response to each \dsop{query},
before responding to the \dsop{query}.
In reality, builds take time.
Can this be modelled cleanly,
perhaps via a problem that constrains the build cost at each time $t$ (and $\wt(\inp_t)$)
to be at most 1, with the objective of minimizing the total query cost?

\paragraph{Splitting the key space.}
To avoid monolithic builds, when the data size reaches some threshold
(e.g., when the available RAM can hold 1\% of the stored data)
some LSM systems ``split'':
they divide the workload into two parts---the keys above and below some threshold---then
restart, handling each part on separate servers.
This requires a mechanism for routing insertions and queries
by key to the appropriate server.
Can this (including a routing layer supporting multiple splits) be cleanly modeled?

Other LSM systems (LevelDB and its derivatives) instead use many small (disk-block size) components,
storing in the (cached) indices each component's key interval (its minimum and maximum key).
A \dsop{query} for a given key accesses only the components whose intervals contain the key.
This suggests a natural modification of our model:
redefine the query cost at time $t$ to be the maximum
number of such components for any key.

\paragraph{Bloom filters.}
Most practical LSM systems are configurable to use a Bloom filter for each component,
so as to avoid (with some probability) accessing component that do not hold the queried key.
However, Bloom filters are only cost-effective when they are small enough to be cached.
They require about a byte per key, so are effective only for the smallest components
(with a total number of keys no more than the bytes available in RAM).
Used effectively, they can save a few disk accesses per \dsop{query}
(see~\cite{dayan_monkey_2017}).
They do not speed up range queries
(that is, efficient searches for all keys in a given interval, 
which LSM systems support but hash-based external-memory dictionaries do not).

\paragraph{External-memory.}
More generally,
to what extent can we apply competitive analysis to the standard I/O (external-memory) model?
Given an input sequence (rather than being constrained to maintain a cover)
the algorithm would be free to use the cache and disk as it pleases, subject only to the constraints
of the I/O model, with the objective of minimizing the number of disk I/O's,
divided by the minimum possible number of disk I/O's for that particular input.
This setting may be too general to work with.
Is there a clean compromise?

The results below do not address this per se,
but they do analyze external-memory algorithms
using metrics other than standard worst-case analysis,
with a somewhat similar flavor:

\begin{description}
\item[\cite{barve_applicationcontrolled_2000}]
  Studies competitive algorithms for allocating cache space to competing processes.
  
\item[\cite{bender_cacheadaptive_2014}]
  Analyzes external-memory algorithms while available RAM varies with time,
  seeking an algorithm such that, no matter how RAM availability varies,
  the worst-case performance is as good as that of any other algorithm.

\item[\cite{brodal_cacheaware_2005}]
  Presents external-memory sorting algorithms
  that have per-input guarantees --- they use fewer I/O's
  for inputs that are ``close'' to sorted.

\item[\cite{ciriani_static_2002,ko_optimal_2007}]
  Present external-memory dictionaries with a kind of \emph{static-optimality} property:
  for any sequence of queries, they incur cost bounded
  in terms of the minimum achievable by any static tree of a certain kind.
  (This is analogous to the static optimality of splay trees~\cite
  {sleator_selfadjusting_1983, levy_new_2019}.)
  
\end{description}

\subsection{Practical considerations}

\paragraph{Heuristics for newest-first solutions.}
Some LSM systems \emph{require} newest-first solutions.
The Min-Sum Dynamization algorithm Adaptive-Binary (Figure~\ref{fig: min-sum alg})
can produce solutions that are not newest-first.
Here is one naive heuristic to make it newest-first: at time $t$, do the minimal newest-first merge
that includes all the components that the algorithm would otherwise have selected to merge.
This might result in only a small cost increase on some workloads.

\paragraph{Major compactions.} 
For various reasons, it can be useful to force major compactions at specified times.
An easy way to model this is to treat each interval between forced major compactions
as a separate problem instance, starting each instance by inserting all items from the major compaction.

\paragraph{Estimating the build cost $\wt_t(S)$.}
Our algorithms for the decreasing-weights, LSM, and general variants
depend on the build costs $\wt_t(S)$ of components $S$ that are not yet built.
The model assumes these become known at time $t$,
but in practice they can be hard to compute.
However, the algorithms only depend on the build costs of components $S$
that are unions of the current components.
For the LSM variant, it may be possible to construct, along with each component $S$,
a small signature that can be used to estimate 
the build costs of unions of such components (at later times $t$),
using techniques for estimating intersections
of large sets~(e.g.~\cite{cohen_minimal_2017,pagh_minwise_2014}).
It would be desirable to show that dynamization algorithms are robust in this
context---that their competitive ratios are approximately preserved
if they use approximate build costs.

\paragraph{Exploiting slack in the \textsf{Greedy-Dual} algorithm.}
For paging, \textsc{Least-Recently-Used} (LRU)
is preferred in practice to \textsc{Flush-When-Full} (FWF),
although their competitive ratios are equal.
In practice, it can be useful to tune an algorithm while preserving its theoretical performance guarantee.
In this spirit, consider the following variant of the \textsf{Greedy-Dual} algorithm
in Figure~\ref{fig: k-component decreasing}.
As the algorithm runs, maintain a ``spare credit'' $\phi$.  Initially  $\phi=0$.
When the algorithm does a merge in Line~\ref{line: merge},
increase $\phi$ by the total credit of the components newer than $S_0$, which the algorithm destroys.
Then, at any time, optionally, reduce $\phi$ by some amount $\delta \le \phi$,
and increase the credit of any component in the cover by $\phi$.
The proof of Theorem~\ref{thm: k-component decreasing}, essentially unchanged,
shows that the modified algorithm is still $k$-competitive.
This kind of additional flexibility may be useful in tuning the algorithm. 
As an example, consider classifying the spare credit by the rank of the component that contributes it,
and, when a new component $S'$ of some rank $r$ is created, 
transferring all spare credit associated with rank $r$ to $\credit[S']$
(after Line~\ref{line: initialize credit} initializes $\credit[S']$ to 0).
This natural \textsc{Balance} algorithm balances the work done for each of the $k$ ranks.

\begin{acks}
  Thanks to Carl Staelin for bringing the problem to our attention
  and for informative discussions about Bigtable.
\end{acks}

%% file: 90_appendix.tex
\appendix

\section{Deferred proofs}

\begin{lemma}\label{lemma: min-sum binary transform}
  For the min-sum dynamization problem, the competitive ratio of the naive adaptation of Bentley’s binary transform, which treats each insertion as a size-1 item and applies the transform, is $\Omega(\log n)$.
\end{lemma}
\begin{proof}
  Consider an input sequence in which an item of weight $n^2$ is inserted in step 1, and an item of weight $\varepsilon$ for infinitesimally small $\varepsilon > 0$ is inserted in each step $t$, for $2 \le t \le n$.
  The naive adaptation of Bentley's binary transform ignores the weights and treats each insertion as a size-1 item.
  Recall that the binary transform maintains at most one component of size $2^i$ for each integer $i$.
  Since the input sequence inserts one item each step, for each $t$ that is a power of two, the binary transform has exactly one component of size $t$ immediately after step $t$ (for instance, see Figure~\ref{fig: binary transform}).
  Thus, each step $t$ that is a power of two incurs build cost at least $n^2$ (owing to the item of weight $n^2$).
  This yields a total build cost of $\Omega(n^2 \log n)$.

  An alternative solution, such as the one computed by the algorithm of Figure~\ref{fig: min-sum alg}, maintains at most two components, one consisting of the weight $n^2$ item and the other consisting of any remaining items.
  The build cost for step 1 is $n^2$.  For step $i$, $2 \le i \le n$, the build cost is $(i-1)\varepsilon$ since $i-1$ items of weight $\varepsilon$ are merged into a component.
  This yields a total build cost of at most $n^2(1 + \varepsilon/2)$.
  Since there are at most two components, the query cost is at most $2n$.
  We thus have an $\Omega(\log n)$ bound on the competitive ratio of the naive adaptation of the binary transform.
\end{proof}

\begin{lemma}\label{lemma: naive bs non-competitive}
  The naive generalization of Bentley and Saxe's $k$-binomial transform to $k$-Component Dynamization
  has competitive ratio $\Omega(k n^{1/k})$ for any $k\ge 2$.
\end{lemma}
\begin{proof}
  Recall that the naive algorithm treats each insertion $\inp_t$ as one size-1 item, then applies the $k$-binomial transform.
  Consider inserting a single item of weight 1, then $n-1$ single items of weight 0.
  The naive algorithm merges its largest component $\Theta(d)$ times where ${d\choose k} \approx n$, so $d = \Theta(k n^{1/k})$.
  Each such merge costs 1.
  So the naive algorithm incurs total cost $\Omega(k n^{1/k})$.

  The optimum keeps the weight-1 item in one component,
  then does all remaining merges into the other (size-zero) component,
  for total cost of 1.
\end{proof}